\documentclass[11pt]{article}
\usepackage{amsmath}
\usepackage{amsthm}
\usepackage{amssymb}
\usepackage{graphicx}
\usepackage{float}
\usepackage{url}
\usepackage{hyperref}
\usepackage[shortlabels]{enumitem}
\usepackage{microtype}
\usepackage{rotating}
\usepackage{booktabs}
\usepackage{geometry}
\usepackage[T1]{fontenc}

\theoremstyle{plain}
\newtheorem{theorem}{Theorem}[section]
\newtheorem{lemma}[theorem]{Lemma}
\newtheorem{corollary}[theorem]{Corollary}

\theoremstyle{definition}
\newtheorem{definition}[theorem]{Definition}
\newtheorem{problem}[theorem]{Problem}
\newtheorem{example}[theorem]{Example}

\theoremstyle{remark}
\newtheorem{remark}[theorem]{Remark}

\newcounter{casenum}
\newenvironment{caseof}{\setcounter{casenum}{1}}{\vskip.5\baselineskip}
\newcommand{\case}[2]{\vskip.5\baselineskip\par\noindent{\bfseries
Case \arabic{casenum}:} #1\\#2\addtocounter{casenum}{1}}

\def\fab{f_{\alpha \beta}}

\usepackage{fancyvrb}
\DefineVerbatimEnvironment{leancode}{Verbatim}{
  fontsize=\small,
  frame=single,
  framesep=3mm,
  rulecolor=\color{gray},
  commandchars=\\\{\}
}
\usepackage{xcolor}
\usepackage[most]{tcolorbox}

\newcommand{\leanbox}[1]{%
  \begin{tcolorbox}[
    enhanced, breakable,
    colback=green!6, colframe=green!45!black,
    fonttitle=\small\sffamily\bfseries,
    title={Lean~4 Certificate},
    left=4pt, right=4pt, top=3pt, bottom=3pt,
    boxrule=0.5pt
  ]
  \small\ttfamily\raggedright #1
  \end{tcolorbox}
}

\newcommand{\leanaxiombox}[1]{%
  \begin{tcolorbox}[
    enhanced, breakable,
    colback=blue!4, colframe=blue!35!black,
    fonttitle=\small\sffamily\bfseries,
    title={Lean~4 Axiom},
    left=4pt, right=4pt, top=3pt, bottom=3pt,
    boxrule=0.5pt
  ]
  \small\ttfamily\raggedright #1
  \end{tcolorbox}
}
\newcommand{\leanfuturebox}[1]{%
  \begin{tcolorbox}[
    enhanced, breakable,
    colback=white, colframe=green!45!black,
    fonttitle=\small\sffamily\bfseries,
    title={Lean~4 --- Future Formalisation},
    left=4pt, right=4pt, top=3pt, bottom=3pt,
    boxrule=0.5pt
  ]
  \small\ttfamily\raggedright #1
  \end{tcolorbox}
}

\title{Lean~4 Machine-Verified Proof of P~=~NP \\ 
via the Pedigree Polytope Membership Problem%
}

\author{T.S.\ Arthanari\\[4pt]
Department of Information Systems and Operations Management\\
University of Auckland, Auckland, New Zealand\\
\texttt{t.arthanari@auckland.ac.nz}\\
ORCID: \href{https://orcid.org/0000-0002-3345-5664}{0000-0002-3345-5664}}

\date{}

\begin{document}
\maketitle

%% Knuth epigraph
\begin{quote}
\begin{center}

\itshape
``Many combinatorial questions that I once thought would never be
answered during my lifetime have now been resolved, and those
breakthroughs have been due mainly to improvements in algorithms
rather than to improvements in processor speeds.''
\upshape
\\  -- Donald~E.\ Knuth,
\textit{The Art of Computer Programming}, Vol.~4A~\cite{knuth2011art}
\end{center}
\end{quote}

\begin{abstract}
The Membership Problem for Pedigree Polytope  (M3P) asks, given
$X\in\mathbb{Q}^{\binom{n}{3}}$, whether $X\in\mathrm{conv}(P_n)$, where $P_n$ is the set of all pedigrees. A pedigree is a structured encoding of a Hamiltonian cycle construction in $K_n$.
We establish that M3P is solvable in strongly polynomial time via a recursively constructed layered network $(N_k, R_k, \mu)$ and a multicommodity flow problem MCF$(k)$.
The necessary and sufficient condition for membership established is that the optimal total flow in MCF$(n-1)$ equals the maximum possible flow $z_{\max}$. The complexity analysis, grounded in Tardos's strongly
polynomial algorithm for combinatorial linear programs (1986),
shows that this condition can be checked in strongly polynomial time in the dimension of the matrix involved.
The MCF$(n-1)$ decision problem --- \emph{is $z^* = z_{\max}$?},
where $z^*$ is the optimal total flow and
$z_{\max} = 1 - \sum_{P \in R_{n-2}}\mu_P$ is the maximum
achievable flow --- is thus in P.
By sufficiency, this implies M3P~$\in$~P.
Since the Symmetric Travelling Salesman Problem (STSP) reduces
to M3P via the Multistage Insertion (MI) formulation
(Arthanari 1983), STSP is solvable in polynomial time, and
the P~vs.~NP question is resolved.
The proofs leading to this result are fully machine-verified in
Lean~4/Mathlib4, with zero unresolved \texttt{sorry}s in the main proof chain:
\begin{center}
\begin{tabular}{@{}l@{\ }l@{}}
\texttt{theorem p\_equals\_np} & \texttt{: P\_class = NP\_class}
\end{tabular}
\end{center}
\noindent
The main contribution is the Lean~4 machine verification of all proofs in the main chain, resulting in \texttt{theorem p\_equals\_np}: P~=~NP.
The Lean~4 formal verification covers the sufficiency of MCF(n-1) for membership in $\mathrm{conv}(P_n)$, and the P~=~NP chain via Maurras (2002), Gr\"{o}tschel--Lov\'{a}sz--Schrijver (1988), Cook (1971), and Karp (1972).
The complete lean project (36 Lean~4 files, 2968/2968 build targets clean) is available at \url{https://github.com/TiruArt/Pedigree-Polytopes-Lean4}.
\end{abstract}

\noindent\textbf{MSC 2020:} 90C27, 03D15\quad
\textbf{ACM:} G.2.1; G.2.2; G.2.3\\[4pt]
\textbf{Keywords:} pedigree polytope, membership problem, strongly
polynomial, multicommodity flow, Lean~4, P~=~NP, STSP
\newpage

%% ============================================================
\section{Introduction}\label{intro}
%% ============================================================

The theory of computational complexity classifies decision problems
by the resources required to solve them.
The class $\mathsf{P}$ consists of problems solvable in polynomial
time by a deterministic algorithm; $\mathsf{NP}$ consists of
problems whose solutions can be \emph{verified} in polynomial time.
Every problem in $\mathsf{P}$ is in $\mathsf{NP}$, but whether
$\mathsf{P}=\mathsf{NP}$ --- whether every efficiently verifiable
problem is also efficiently solvable --- has been the central
open question in computer science since Cook~\cite{CookNP} and
Karp~\cite{Karp} established the theory of NP-completeness in the early 1970s.

A problem is \emph{NP-complete} if it is in $\mathsf{NP}$ and
every problem in $\mathsf{NP}$ reduces to it in polynomial time.
NP-complete problems are, in a precise sense, the hardest problems
in $\mathsf{NP}$: if any one of them is in $\mathsf{P}$, then
$\mathsf{P}=\mathsf{NP}$ and every problem in $\mathsf{NP}$ becomes
efficiently solvable.

The \emph{Symmetric Travelling Salesman Problem} (STSP) --- find
the shortest Hamiltonian cycle through $n$ cities --- is one of the
most intensively studied NP-complete problems~\cite{Karp,GandJ}.
Its decision version (is there a tour of length $\leq L$?) is
NP-complete; its optimisation version is NP-hard.
This paper proves that STSP is solvable in polynomial time, via
a strongly polynomial algorithm for a polyhedral membership problem.
The consequence --- $\mathsf{P}=\mathsf{NP}$ --- is
machine-verified in Lean~4, destroying the boundary between
the complexity classes $\mathsf{P}$ and $\mathsf{NP}$:
every problem whose solution can be verified efficiently can
also be \emph{found} efficiently.

Let $n\geq 3$ and $V_n=\{1,\ldots,n\}$. Let
$E_n=\{(i,j)\mid i,j\in V_n,\,i<j\}$ with $|E_n|=p_n=n(n-1)/2$.
Let $\Delta^k=\{\{i,j,k\}\mid 1\leq i<j<k\}$ for $k\in[3,n]$.
For $k >3$, a \emph{generator} of $v=\{i,j,k\}$ is a triangle $u=\{r,i,j\}$
or $u=\{i,s,j\}$ (for $j>3$), or $u=\{1,2,3\}$ (for $j=3$). 
The edge $\{i,j\}$ is the \emph{common edge} of $v$. The triangle 
$u=\{1,2,3\}$ is called the \emph {base triangle}, and has no generators.

\begin{definition}[Pedigree~\cite{TSADMpaper}]\label{pedig}
A \emph{pedigree} $P$ is a sequence of $n-2$ triangles
$$P=(\{1,2,3\},\{i_4,j_4,4\},\ldots,\{i_n,j_n,n\})$$
such that (1) each $\{i_k,j_k,k\}$ has a generator in $P$ for
$k\in[4,n]$, and (2) the common edges $\{i_4,j_4\},\ldots,\{i_n,j_n\}$
are all distinct.
\end{definition}

The set of characteristic vectors of pedigrees for $n$
(omitting the last coordinate per layer) is $P_n$.
The \emph{pedigree polytope} is $\mathrm{conv}(P_n)$.
We write $\tau_n=\binom{n}{3}-1$ for the number of 
triangles, when we omit $\{1,2,3\}$.
The dimension of $\mathrm{conv}(P_n)$ is $\tau_n-(n-3)$~\cite{TSACompoly,arthanari2023pedigree}.

\begin{problem}[M3P]\label{problemain}
\emph{Input}: $n\geq 4$, $X\in\mathbb{Q}^{\binom{n}{3}}$.
\emph{Question}: Is $X\in\mathrm{conv}(P_n)$?
\end{problem}

Henceforth, \emph{the book} refers to \emph{Pedigree Polytopes: New Insights on Computational Complexity of Combinatorial Optimisation Problems}~\cite{arthanari2023pedigree}.
Our main result, that M3P is solvable in strongly polynomial time, which is from Chapter~5 of the book, was shared in the arXiv paper~\cite{arthanari2025arXiv1}, with some new proofs and simplifications (the principal purpose of the arXiv article is to invite experts to vet the proofs and constructions).
However, a consequence of that result, as proved in Chapter~7 of the book, implies P~=~NP.
The current article explains and demonstrates how the P~=~NP consequence is machine-verified in Lean~4 (Section~\ref{sec:lean4}).
The chain of lemmas, theorems and constructions leading to the final theorem is depicted below:
\begin{enumerate}
\item \textbf{Chapter~5} (proved): MCF$(n-1)$ feasible with $z^*=z_{\max}$
      $\Rightarrow$ $X\in\mathrm{conv}(P_n)$.
\item \textbf{Chapter~6} (Tardos~\cite{StrPoly}): MCF$(n-1)$ is a
      combinatorial LP $\Rightarrow$ the decision problem
      \emph{``is $z^*=z_{\max}$?''} is in P
      $\Rightarrow$ M3P~$\in$~P (strongly polynomial).
\item \textbf{Chapter~7} (proved): $\mathrm{conv}(A_n)$ full dimensional
      $+$ rationality guaranteed $+$ $a\in\mathrm{int}(\mathrm{conv}(A_n))$ $+$ M3P$\in$P
      $\Rightarrow$ polynomial optimisation over $\mathrm{conv}(P_n)$
      (Maurras~\cite{Maurras}, GLS~\cite{GLS}).
      The polytope $\mathrm{conv}(A_n)$ is a related polytope to $\mathrm{conv}(P_n)$.
\item \textbf{Chapter~3} (Arthanari~\cite{TSA1}): STSP reduces to the MCF$(n-1)$ decision problem via the MI-formulation
      $\Rightarrow$ STSP is solvable in polynomial time.
\item \textbf{Karp~\cite{Karp}, Cook~\cite{CookNP}}: STSP decision
      NP-complete $+$ STSP optimisation~$\in$~P $\Rightarrow$ P~=~NP.
\end{enumerate}

\noindent The companion paper~\cite{arthanari2025arXiv1} provides the motivation for studying M3P and the broader context.
The present paper gives a complete strongly polynomial algorithm, the N\&S characterisation, and the Lean~4 machine-verified P~=~NP consequences.

\medskip
The framework presented in this paper draws on a lineage of
ideas from combinatorial optimisation. Working in the world of
\emph{triangles} --- multistage insertion of cities into growing partial tours ---
gives the pedigree polytope its recursive structure.
The key discipline, articulated by Edmonds and reinforced by
Tardos is the elimination of non-determinism at every step.

The matrix of the multicommodity flow MCF$(k)$ has entries in
$\{0,\pm1\}$ --- making it a combinatorial LP and guaranteeing
Tardos's strongly polynomial algorithm is applicable.

\medskip
Most approaches to solving the STSP focus on the
Dantzig--Fulkerson--Johnson (DFJ) formulation~\cite{DFJ}
and the subtour elimination polytope (SEP).
The SEP has an exponential number of constraints, but the
cutting-plane method --- separating subtour elimination
constraints via max-flow --- combined with branch and bound/cut
has produced extraordinarily powerful solvers.
The \textit{Concorde} TSP solver~\cite{Applegate2006},
built on this foundation, has solved instances with tens of
thousands of cities to optimality, and remains the
state of the art.
Other significant formulations include the Held--Karp
assignment relaxation~\cite{HeldKarp1970}, the Christofides
$\frac{3}{2}$-approximation~\cite{Christofides1976},
and the Padberg--Rinaldi branch-and-cut framework~\cite{PadbergRinaldi1991}.
(In Chapter~8 of the book, many other formulations are mentioned and compared with the MI-
formulation, on problems of size
up to 300.)

The pedigree polytope approach is fundamentally different:
it works in the space of triangles rather than edges,
uses the MI-relaxation rather than SEP, and its membership
problem is solvable in strongly polynomial time without
cutting planes or branch and bound, using only additions
and subtractions.
Dantzig, Fulkerson and Johnson's 42-city problem~\cite{DFJ}
--- the first large-scale TSP solved by integer programming
methods --- is used in this paper to illustrate the M3P
membership check (Section~\ref{sec:simplex}).

The paper is organised as follows. Section~\ref{preli} gives
preliminaries. Section~\ref{mi_formulation} introduces the MI-formulation.
Section~\ref{conslay} constructs the layered network.
Section~\ref{Fkfeasibility} proves the sufficient condition for non-membership.
Section~\ref{multi} defines MCF$(k)$ and proves the N\&S theorem.
Section~\ref{compcomplex} analyses complexity.
Section~\ref{sec:lean4} presents the Lean~4 machine verification and
P~=~NP chain.
Section~\ref{sec:newbeginning} (``A New Beginning'') summarises
the M3P framework, opens Pandora's box of consequences for the
``unless $\mathsf{P}=\mathsf{NP}$'' literature, shows the
algebraic topology connection, and lays out the new research agenda.
There is no concluding section: the paper ends with a new beginning.

%% ============================================================
\section{Preliminaries \& Notation}\label{preli}
%% ============================================================

\noindent\textbf{Guide to coloured boxes.}
Throughout the paper, each theorem, lemma, and corollary is
annotated with a coloured box indicating its Lean~4 status,
following the blueprint convention~\cite{lean4,mathlib4}:
\begin{itemize}\setlength{\itemsep}{2pt}
\item \textbf{Green box} (\textit{Lean~4 Certificate}):
  the result is fully formalised and verified in Lean~4
  with zero \texttt{sorry}s.
\item \textbf{White box with green border}
  (\textit{Lean~4 --- Future Formalisation}):
  the statement is formalised but the proof is a planned
  future project.
\item \textbf{Blue box} (\textit{Lean~4 Axiom}):
  the result is used as an \texttt{axiom} declaration ---
  an accepted published theorem assumed without re-proof.
\end{itemize}

\medskip
We follow standard notation: $\mathbb{R},\mathbb{Q},\mathbb{Z},\mathbb{N}$
denote the reals, rationals, integers and natural numbers; $B=\{0,1\}$;
subscript ${}_+$ denotes nonnegativity.

\subsection{Network Flow Problems and Combinatorial LPs}

The algorithmic framework of this paper rests on a chain of
flow problems whose roots lie in the analogy between
transportation and electrical networks, observed by
Koopmans~\cite{koopmans1951model}. As Koopmans noted, there is
a striking correspondence between shipping surplus at a node
and net current injected into an electrical network ---
a correspondence that connects Kirchhoff's 1847 electrical
network theory to modern transportation problems.

Ford and Fulkerson~\cite{FF} formalised this as the
\emph{max-flow problem}: given a network with arc capacities,
find a maximum flow from source to sink. Their labelling
method, augmenting along paths in the residual network, gives
the foundation for all flow algorithms.
The strongly polynomial algorithm of Edmonds and Karp~\cite{EdmondsKarp}
is used in this paper.

\begin{definition}[Combinatorial LP~\cite{StrPoly}]\label{combLP}
A class of LP problems is \emph{combinatorial} if the entry sizes
of the matrix $A$ are polynomially bounded in the problem dimension.
Tardos~\cite{StrPoly} gives a strongly polynomial algorithm for
combinatorial LPs whose number of arithmetic steps depends only on
$\dim(A)$, independent of $b$ and $c$.
\end{definition}

The M3P framework also relies on the
\emph{Forbidden Arc Transportation (FAT) problem} --- a
bipartite flow problem where some arcs are prohibited.
The FAT problem arises at each stage $k$ as $F_k$.

\subsection{Rigid and Dummy Arcs in a Transportation Problem}

\begin{definition}[FAT Problem]
A \emph{Forbidden Arcs Transportation (FAT) problem} is a balanced
transportation problem with some arcs forbidden. Given origins
$O=\{O_\alpha\}$ with supplies $a_\alpha\geq0$, destinations
$D=\{D_\beta\}$ with requirements $b_\beta\geq0$,
$\sum_\alpha a_\alpha=\sum_\beta b_\beta$, and permitted arcs
$\mathcal{A}\subseteq O\times D$, the FAT problem seeks a feasible
non-negative flow satisfying all supply/demand constraints.
\end{definition}

\begin{definition}[Rigid Arc]
Given a FAT problem with feasible solution $f$, arc $(\alpha,\beta)$
is \emph{rigid} if $f_{\alpha\beta}$ is the same in every feasible
solution. A rigid arc with zero frozen flow is a \emph{dummy} arc.
\end{definition}

The set of rigid arcs can be identified in linear time
$O(|G_f|)$ using the \emph{Frozen Flow Finding (FFF) algorithm}
of Gusfield~\cite{Gus}; see Appendix~\ref{Appendixfrozen}.

\begin{lemma}[Flow Partition]\label{theorem:lemma3}
Let $\mathcal{D}\neq\emptyset$, $g:\mathcal{D}\to\mathbb{Q}_+$,
and $\mathcal{D}^1,\mathcal{D}^2$ be two non-empty partitions of
$\mathcal{D}$. The FAT problem with origins $D^1_\alpha$
(supply $g(D^1_\alpha)$), destinations $D^2_\beta$
(demand $g(D^2_\beta)$), and arcs
$\mathcal{A}=\{(\alpha,\beta)\mid D^1_\alpha\cap D^2_\beta\neq\emptyset\}$
has feasible solution $\fab=g(D^1_\alpha\cap D^2_\beta)$.
\end{lemma}
\leanbox{\texttt{N\_LayeredNetworkTypes.lean} \textbar{} \texttt{lemma flow\_partition}: arc capacities sum to node capacity \hfill\textit{(book, Chapter~2)}}

\begin{example}\label{RigidEx}
[a] Consider the FAT problem with $O=\{1,2,3\}$,
$a(1)=0.3$, $a(2)=0.3$, $a(3)=0.4$; $D=\{4,5,6,7\}$,
$b(4)=0.4$, $b(5)=0.2$, $b(6)=0.1$, $b(7)=0.3$; forbidden arcs
$F=\{(1,6),(1,7),(3,4),(3,6)\}$.
Since $(2,6)$ is the only arc entering destination~6, it is rigid
with frozen flow $f_{2,6}=0.1$. See Figure~\ref{RigidFig}.

[b] Adding capacity $c_{1,4}=2$: all arcs become rigid, and
$(2,5)$, $(2,7)$ are dummy arcs. See Figure~\ref{RigidFig_b}.
\end{example}

\begin{figure}[htb]
  \centering
  \includegraphics[width=0.78\textwidth]{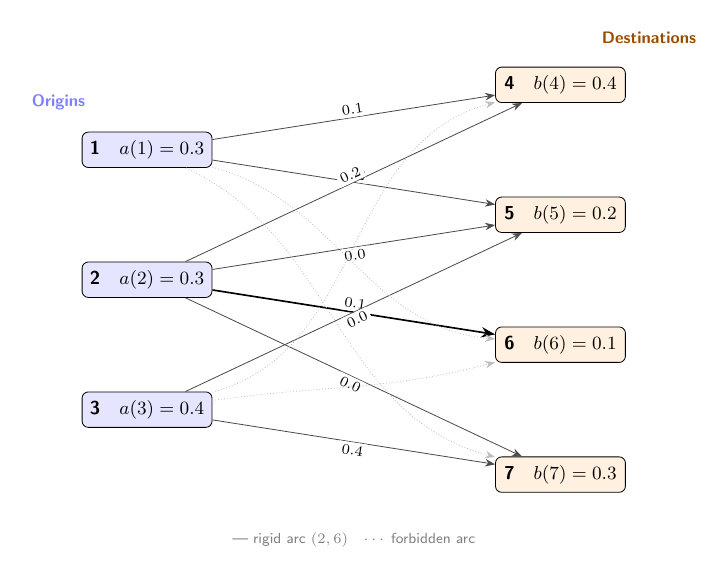}
  \caption{FAT problem for Example~\ref{RigidEx}[a].
  Flows shown along arcs; $a(i)$ and $b(j)$ in boxes.}
  \label{RigidFig}
\end{figure}

\begin{figure}[htb]
  \centering
  \includegraphics[width=0.78\textwidth]{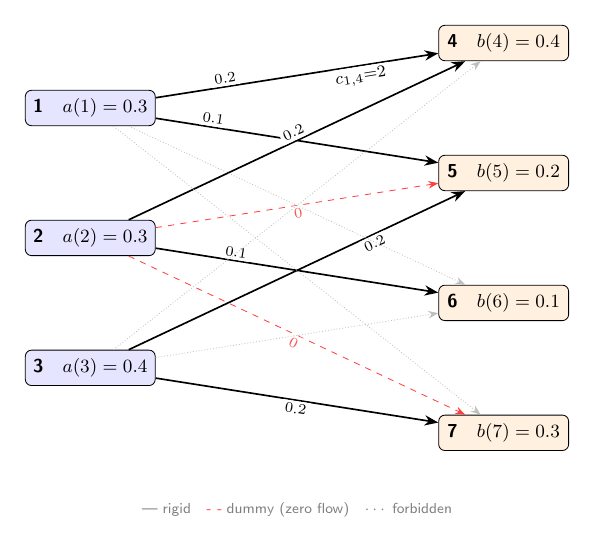}
  \caption{FAT problem for Example~\ref{RigidEx}[b].
  Rigid arcs in black, dummy arcs in red.}
  \label{RigidFig_b}
\end{figure}

%% ============================================================
\section{The MI-Formulation and the Pedigree Polytope}\label{mi_formulation}
%% ============================================================

\subsection{Pedigree Optimisation Problem}

For each triangle $u=\{i,j,k\}\in\Delta^k$, we have an $c_{ijk}$, the objective coefficient for $u$, and  let $x_{ijk}=1$ if $u\in P$, $0$ otherwise. 
The \emph{Pedigree Optimisation Problem}~\cite{TSA1} is:

\begin{problem}[Pedigree Optimisation]\label{pop}
\begin{align}
\min\quad&\sum_{u\in\Delta}\mathcal{C}_u x_u \label{eq:obj}\\
\text{subject to}\quad
\sum_{u\in\Delta^k}x_u &= 1,\quad k\in[3,n],\label{eq:6}\\
\sum_{k\in[4,n]}x_{ijk} &\leq 1,\quad (i,j)\in E_3,\label{eq:7}\\
-\sum_{u\in\Delta^j,\,i\in u}x_u + \sum_{k\in[j+1,n]}x_{ijk}
&\leq 0,\quad (i,j)\in E_{n-1}\setminus E_3,\label{eq:8}\\
x_u &\in\{0,1\},\quad k\in[4,n],\,u\in\Delta^k.\label{eq:9}
\end{align}
\end{problem}

Let $c_{ij}$ be the symmetric distance from vertex $i$  to vertex $j$, given for an instance of the STSP.  Starting with the $3-tour, 1-2-3-1$, we can insert $k=4,\ldots,n$, sequentially and obtain an $n$-tour. This multistage insertion idea is used in the $MI$-formulation of the STSP \cite{TSA1, Alt}. Observe that with objective coefficients $\mathcal{C}_{ijk}=c_{ik}+c_{jk}-c_{ij}$ (incremental insertion cost), Problem~\ref{pop} is the
\emph{MI-formulation}~\cite{TSA1} of the STSP.
Thus, the pedigree optimisation problem is a larger class of
problems that includes a formulation that solves the STSP problem.
Replacing~\eqref{eq:9} by $x_u\geq0$ gives the \emph{MI-relaxation}
$MIR(n)$ with feasible set $P_{MI}(n)$.

\begin{lemma}[Lemma~1~\cite{Alt}]\label{oneone}
Every integer solution $X$ to $MIR(n)$ has slack vector
$\mathbf{u}\in B^{p_n}$ equal to the edge-tour incident vector of
the corresponding $n$-tour.
\end{lemma}
\begin{proof}
For any integer solution to $MIR(n)$:
[1]~the slack for~\eqref{eq:7} is $0$ or $1$ according to
$\sum_{k\in[4,n]}x_{ijk}\in\{1,0\}$;
[2]~the left side of~\eqref{eq:8} is $0$ when
$\sum_{u\in\Delta^j,\,i\in u}x_u=1$ and $\sum_{k\in[j+1,n]}x_{ijk}=1$;
[3]~it is $0$ when $\sum_{u\in\Delta^j,\,i\in u}x_u=0$.
Given $X$ with $x_{i_kj_kk}=1$ for $k\in[4,n]$: start from the
3-tour on $E_3$; at each stage inserting $k$ in edge $(i_k,j_k)$
replaces it and creates $(i_k,k),(j_k,k)$, giving a $k$-tour.
Whenever an available edge is used for insertion, or an edge is not
created by any insertion decision, the corresponding slack $u_{ij}=0$.
So $\mathbf{u}$ is the edge-tour incidence vector of the $n$-tour.
\end{proof}
\leanbox{\texttt{N\_Basic.lean} \textbar{} \texttt{theorem oneone}: bijection between pedigrees and integer MIR solutions \hfill\textit{(book, Chapter~3)}}

\begin{theorem}[Packability Corollary]\label{packabilitycor}
If $X\in P_{MI}(n)$ and $X/k\in\mathrm{conv}(P_k)$, and
$x_{k+1}(e')=1$ for some $e'$, then $X/k{+}1\in\mathrm{conv}(P_{k+1})$.
\end{theorem}
\begin{proof}
Since $X/k\in\mathrm{conv}(P_k)$, choose $\lambda\in\Lambda_k(X)$.
For each $r\in I(\lambda)$, $X^r$ is an integer solution to $MIR(k)$
with edge-tour incidence vector $U^r\in B^{p_k}$ by Lemma~\ref{oneone}.
Since $X/k=\sum_r\lambda_rX^r$, the slack $\mathbf{u}=\sum_r\lambda_rU^r$.
Since $X/k{+}1\in P_{MI}(k{+}1)$, we have $x_{k+1}\leq\mathbf{u}$, so
$x_{k+1}(e')=1\Rightarrow u_{e'}=1\Rightarrow e'\in T^r$ for all
$r\in I(\lambda)$.
Define $Y^r=(X^r,\mathrm{ind}(\{e',k{+}1\}))$; this is a pedigree
in $P_{k+1}$ for each $r$, and
$X/k{+}1=\sum_r\lambda_rY^r\in\mathrm{conv}(P_{k+1})$.
\end{proof}
\leanbox{N\_PackabilityCorollary.lean\\
theorem packability\_corollary \{n k : N\} (hk : 4 <= k) ...\\
\phantom{xx}: X/(k+1) in conv(P\_(k+1))
\textit{[new proof, Chapter 5, Lean~4 verified]}}

\subsection{Multistage Insertion and Related Results}\label{MI_Ins}

\begin{example}
For $n=5$, $X'=(1/2,0,1/2;\;1/2,0,0,0,0,1/2)\in\mathbb{Q}^{\tau_5}$
satisfies $MIR(5)$ with $U^{(1)}=(1/2,1,1/2,1/2,1,1/2)'$ and
$U^{(2)}=(0,1,1/2,1/2,1,0,1/2,1/2,1/2,1/2)$.
The matrices $E_{[5]}$ and $A_{[5]}$ are shown in
Figure~\ref{matrices}.
\end{example}

\begin{figure}[ht]
\[
E_{[5]}=\begin{pmatrix}
1&1&1&0&0&0&0&0&0\\
0&0&0&1&1&1&1&1&1
\end{pmatrix},
\]
\[
A_{[5]}=\begin{pmatrix}
 1& 0& 0& 1& 0& 0& 0& 0& 0\\
 0& 1& 0& 0& 1& 0& 0& 0& 0\\
 0& 0& 1& 0& 0& 1& 0& 0& 0\\
-1&-1& 0& 0& 0& 0& 1& 0& 0\\
-1& 0&-1& 0& 0& 0& 0& 1& 0\\
 0&-1&-1& 0& 0& 0& 0& 0& 1\\
 0& 0& 0&-1&-1& 0&-1& 0& 0\\
 0& 0& 0&-1& 0&-1& 0&-1& 0\\
 0& 0& 0& 0&-1&-1& 0& 0&-1\\
 0& 0& 0& 0& 0& 0&-1&-1&-1
\end{pmatrix}.
\]
\caption{Matrices $E_{[5]}$ and $A_{[5]}$.}\label{matrices}
\end{figure}

The MI-relaxation is a combinatorial LP (Definition~\ref{combLP}) of
dimension less than $n^2\times n^3$.

%% ============================================================
\section{Construction of the Layered Network}\label{conslay}
\medskip
\noindent\textbf{All proofs in this section are from the book unless stated otherwise.}
\medskip

%% ============================================================
The paper~\cite{TSADMpaper} establishes that membership of $X$ in
$\mathrm{conv}(P_n)$ can be decided by checking the feasibility of
a sequence of FAT problems $FAT_k(\lambda)$ for $k=4,\ldots,n-1$.
However, finding a suitable $\lambda_k$ at each stage requires
knowledge of the convex combination expressing $X/k$ in $P_k$, and
the complexity of determining such a $\lambda_k$ was not addressed.
The layered network construction resolves this: it generates the
necessary weight information as a by-product of flow computations,
without requiring an explicit search over $\Lambda_k(X)$.

Figure~\ref{fig:layers} shows the layer structure of the network
$N_k$ in full generality.
Each screen represents one layer: the nodes within it are the
edges of $K_k$ with positive weight in $X$, each carrying
capacity $x_k(e)$.
The spine arrows show city $k$ being inserted into edge $e$,
advancing the construction from layer $k-1$ to layer $k$.

\begin{figure}[ht]
  \centering
  \includegraphics[width=\textwidth]{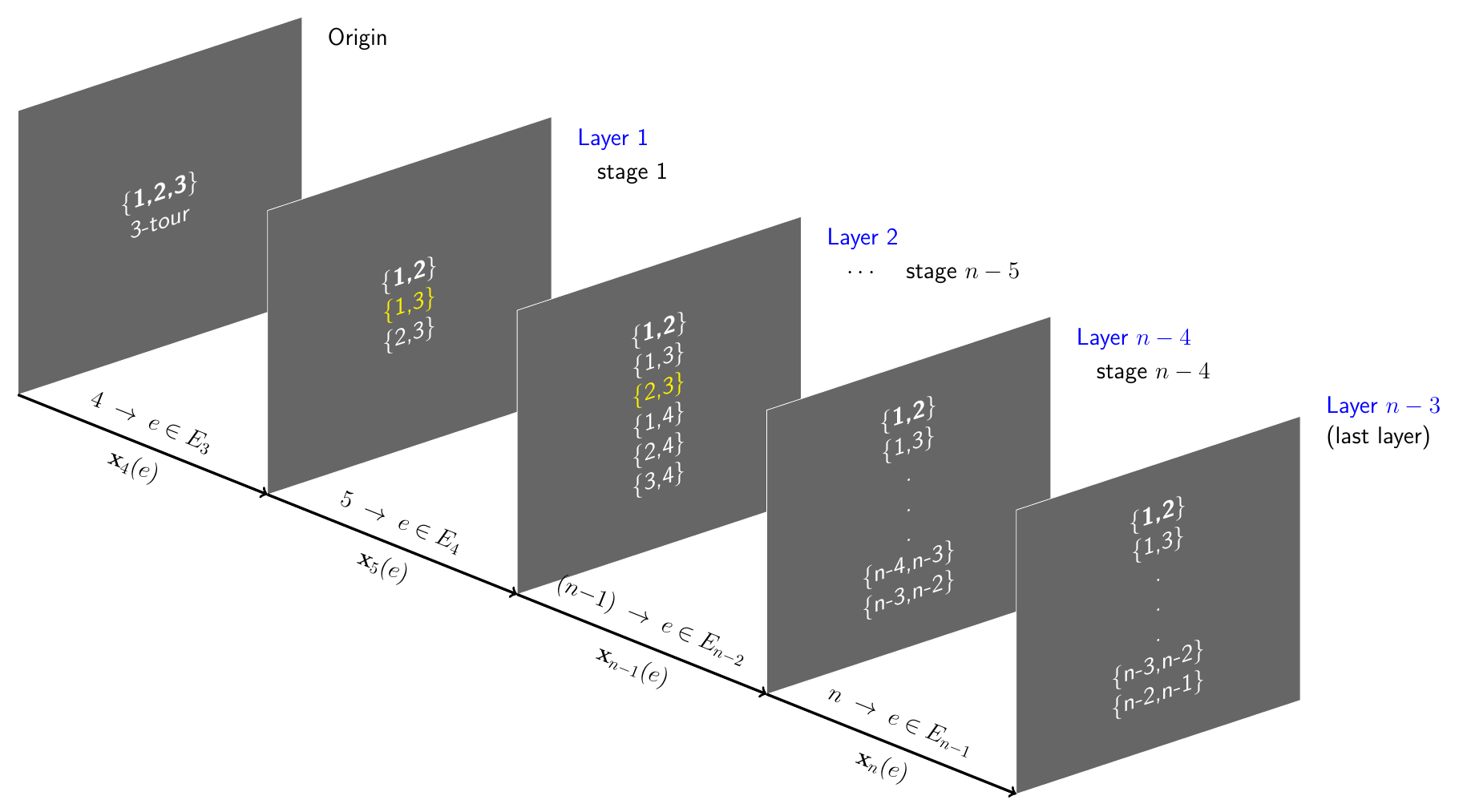}
  \caption{The layers of the network $N_k$: nodes in each screen
  are the edges of $K_k$ with capacity $\mathbf{x}_k(e)$; spine
  arrows show city $k$ inserted into edge $e$.}
  \label{fig:layers}
\end{figure}
\subsection{A Motivating Example: Checking
$X\in\mathrm{conv}(P_7)$}\label{sec:motivating}

Before developing the general construction, we illustrate the key
idea with a small example.
For $n=7$ cities, consider the point
$X=(\mathbf{x}_4,\mathbf{x}_5,\mathbf{x}_6,\mathbf{x}_7)\in P_{MI}(7)$
with the following layer components (edges of $K_k$ in
lexicographic order; zero entries omitted):
\begin{align*}
\mathbf{x}_4 &= \bigl(\tfrac{1}{2},\,\tfrac{1}{4},\,\tfrac{1}{4}\bigr)
  \quad\text{over edges }\{1,2\},\{1,3\},\{2,3\},\\
\mathbf{x}_5 &= \bigl(\tfrac{1}{4},\,\tfrac{1}{2},\,\tfrac{1}{4}\bigr)
  \quad\text{over edges }\{1,2\},\{1,3\},\{2,3\}
  \quad\text{(other edges zero)},\\
\mathbf{x}_6 &= \bigl(\tfrac{1}{4},\,\tfrac{1}{4},\,\tfrac{1}{2}\bigr)
  \quad\text{over edges }\{1,2\},\{1,3\},\{2,3\}
  \quad\text{(other edges zero)},\\
\mathbf{x}_7 &= \bigl(\tfrac{1}{2},\,\tfrac{1}{4},\,\tfrac{1}{4}\bigr)
  \quad\text{over edges }\{1,4\},\{2,4\},\{3,4\}
  \quad\text{(other edges zero)}.
\end{align*}

\noindent\textbf{Question:} Is $X\in\mathrm{conv}(P_7)$?

\medskip
The M3P algorithm builds the layered network $N_6$ from $X$
layer by layer.
Each layer corresponds to one city insertion: the nodes at layer~$k$
are the edges of $K_k$ with positive $x_k(e)$, each carrying
capacity $x_k(e)$.
At layer~1 (city~4), the three nodes $\{1,2\}$, $\{1,3\}$, $\{2,3\}$
carry capacities $\tfrac{1}{2}$, $\tfrac{1}{4}$, $\tfrac{1}{4}$
respectively; and so on for layers 2, 3, 4.
The membership question reduces to: \emph{does MCF$(6)$ admit a
feasible flow with total value $z^*=z_{\max}$?}

\medskip
\noindent\textbf{Answer:} Yes --- the flow is shown in
Figure~\ref{fig:flow-p7}.
It decomposes into three pedigree paths, one per colour:
\begin{align*}
P_1\ (\lambda=\tfrac{1}{2}):\quad
  &\{1,2\}\to\{1,3\}\to\{2,3\}\to\{1,4\},\\
P_2\ (\lambda=\tfrac{1}{4}):\quad
  &\{2,3\}\to\{1,2\}\to\{1,3\}\to\{2,4\},\\
P_3\ (\lambda=\tfrac{1}{4}):\quad
  &\{1,3\}\to\{2,3\}\to\{1,2\}\to\{3,4\}.
\end{align*}
At every node, the total flow in equals the total flow out and
equals the node capacity $x_k(e)$; the flows sum to
$z^*=z_{\max}=1$.
This is precisely the MCF feasibility condition of
Theorem~\ref{maintheorem}, certifying
\[
X = \tfrac{1}{2}X_{P_1}+\tfrac{1}{4}X_{P_2}+\tfrac{1}{4}X_{P_3}
\;\in\;\mathrm{conv}(P_7).
\]

\begin{figure}[ht]
  \centering
  \includegraphics[width=\textwidth]{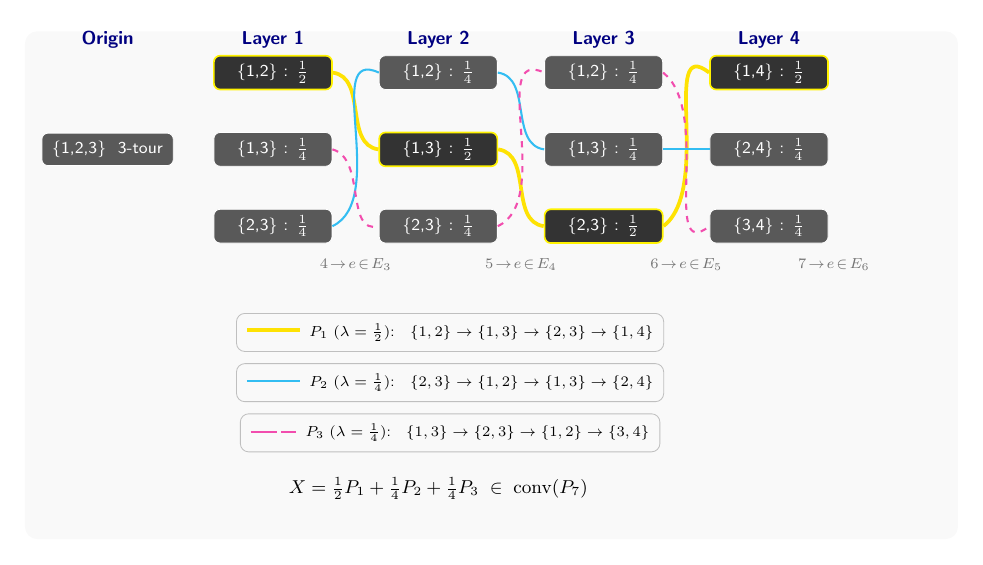}
  \caption{Layered network $N_6$ certifying $X\in\mathrm{conv}(P_7)$.
  Each node shows its edge and capacity $x_k(e)$.
  The three coloured paths exhibit $X$ as a convex combination
  of pedigrees $P_1$, $P_2$, $P_3$ (see legend) ---
  this is the proof of membership.}
  \label{fig:flow-p7}
\end{figure}

But in general, given an arbitrary $X\in P_{MI}(n)$, we do not
know in advance whether $X\in\mathrm{conv}(P_n)$, nor do we have an explicit convex combination to hand.
The central question is: \emph{how do we check membership
efficiently?}
The layered network construction answers this by building the
network $N_k$ recursively, layer by layer, solving a sequence of
flow problems whose feasibility collectively determines membership.

\subsection{Construction of the Network for $k=4$}\label{consfour}

Let $V_{[r]}=\{v\mid v=[r+3:e],\,e\in E_{r+2},\,x(v)>0\}$.
The network $N_4$ has node set $\mathcal{V}(N_4)=V_{[1]}\cup V_{[2]}$
and arc set
$$\mathcal{A}(N_4)=\{(u,v)\mid u=[4:e_\alpha]\in V_{[1]},\,
v=[5:e_\beta]\in V_{[2]},\,e_\alpha\in G(e_\beta)\},$$
where $G(e_\beta)$ denotes the set of generators of $e_\beta$
(Definition~\ref{pedig}).

Arc $(u,v)\in\mathcal{A}(N_4)$ carries capacity $x(u)$;
node $w$ has capacity $x(w)$.

We solve the FAT problem $F_4$ on $N_4$. If $F_4$ is infeasible,
$X\notin\mathrm{conv}(P_n)$. Otherwise, the FFF algorithm identifies
rigid arcs, giving the rigid set $R_4$ with weights $\mu_P$, and
we obtain the well-defined triple $(N_4,R_4,\mu)$.

\begin{example}\label{exampleFat}
Let $X=(0,\frac{1}{3},\frac{2}{3},0,\frac{1}{6},0,
\frac{1}{6},\frac{1}{3},\frac{1}{3})$.
We check $X\in\mathrm{conv}(P_5)$. We have
$\Lambda_4(X)=\{(0,\frac{1}{3},\frac{2}{3})\}$ and $I(\lambda)=\{2,3\}$.
The FAT problem $F_4$ has origins $\{2,3\}$ with supplies
$\frac{1}{3},\frac{2}{3}$ and destinations $\{2,4,5,6\}$.
The feasible flow
$f_{24}=f_{26}=f_{32}=f_{36}=\frac{1}{6}$, $f_{35}=\frac{1}{3}$
certifies $X\in\mathrm{conv}(P_5)$ via five pedigrees, each rigid.
See Figure~\ref{figure1}.
\end{example}

\begin{figure}[htb]
  \centering
  \includegraphics[width=0.82\textwidth]{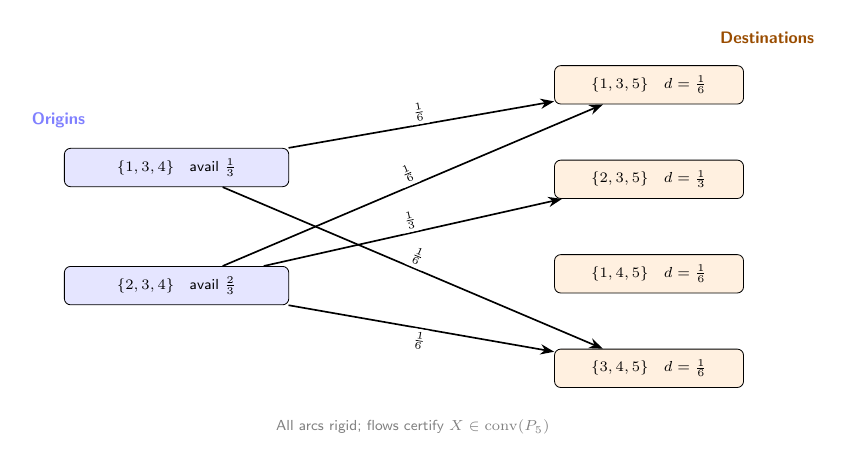}
  \caption{$FAT_4(\lambda)$ for Example~\ref{exampleFat}.
  Origin nodes show the triangle (pedigree) and availability;
  destination nodes show the triangle and demand.
  All arcs are rigid; flows certify $X\in\mathrm{conv}(P_5)$.}
  \label{figure1}
\end{figure}

\subsection{Overview of the Membership Checking Framework}

The full algorithmic framework is shown in Figure~\ref{fig1}.
If $(N_{k-1},R_{k-1},\mu)$ is well-defined and $k<n$, we construct
$F_k$: if infeasible, $X\notin\mathrm{conv}(P_n)$; if feasible,
we construct $(N_k,R_k,\mu)$ and solve MCF$(k)$.

\begin{figure}[ht]
  \centering
  \includegraphics[scale=0.82]{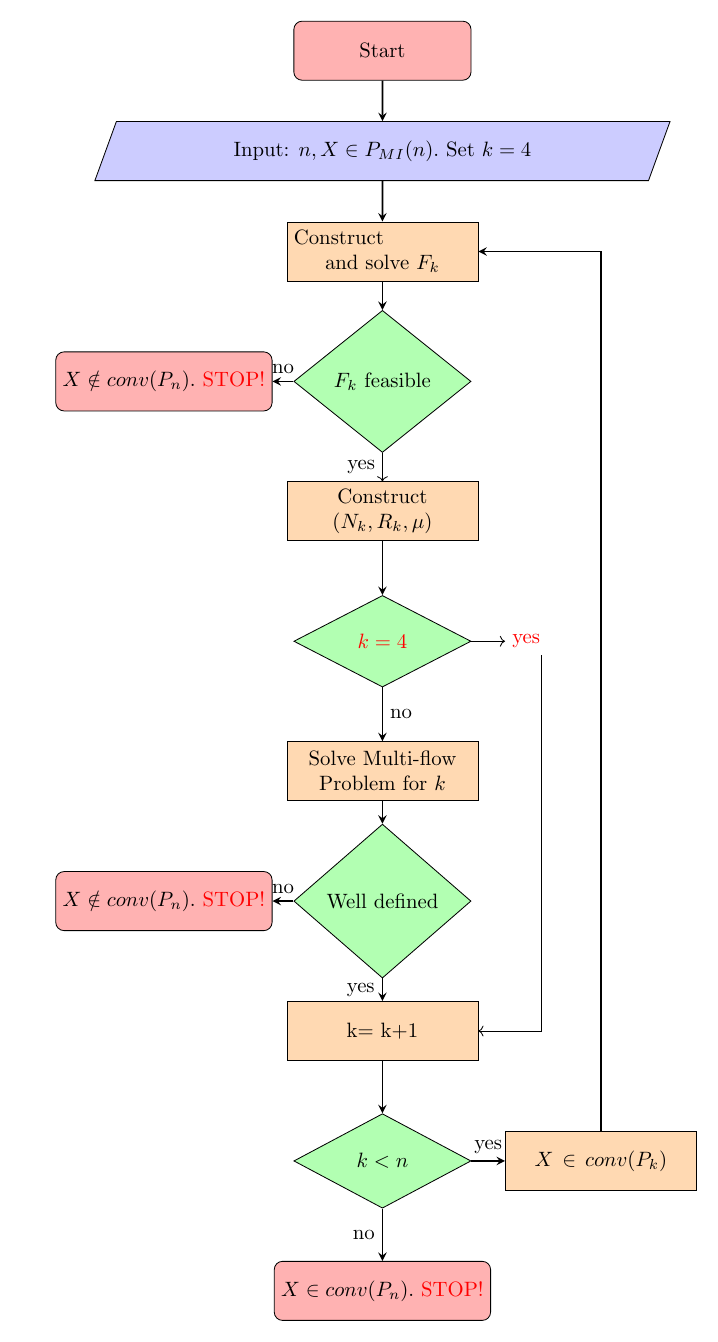}
  \caption{Flowchart for membership checking in $\mathrm{conv}(P_n)$.}
  \label{fig1}
\end{figure}

\subsection{Construction of the Layered Network: $N_k$, $k>4$}

\begin{definition}[Restricted Network $N_{k-1}(L)$]\label{restrict}
Given $k\in[5,n-1]$ and link $L=(u=\{r,s,k\},v=\{i,j,k+1\})$,
$N_{k-1}(L)$ is the subnetwork induced by $\mathcal{V}(N_{k-1})\setminus\mathbf{D}$,
where $\mathbf{D}$ is constructed by deletion rules (a)--(g):
(a) include $\{i,j,l\}$ for $l\in[\max(4,j),k-1]$;
(b) include $\{r,s,l\}$ for $l\in[\max(4,s),k-1]$;
(c) include $w\in\Delta^s$, $w\notin G(u)$;
(d) include $w\in\Delta^j$, $w\notin G(v)$;
(e) include all nodes in $\Delta^k\setminus\{u\}$;
(f) cascade: include any undeleted node whose every generator has
been deleted, until no more deletions occur;
(g) include rigid pedigrees containing any deleted node.
\end{definition}

The construction of $F_k$ ($k>4$) uses nodes in $V_{[k-3]}$
plus one \emph{shrunk} node per $P\in R_{k-1}$, with
capacity $C(L)=\text{max-flow in }N_{k-1}(L)$ for each link $L$.

\begin{example}\label{exsix}
Let $X\in P_{MI}(6)$ with
$\mathbf{x}_4=(0,\tfrac{3}{4},\tfrac{1}{4})$,
$\mathbf{x}_5=(\tfrac{1}{2},0,0,0,0,\tfrac{1}{2})$, and
\[
\mathbf{x}_6=(0,\tfrac{1}{4},\tfrac{1}{4},0,\tfrac{1}{4},
\tfrac{1}{4},0,0,0,0).
\]
$F_4$ and $F_5$ are feasible, $R_4=R_5=\emptyset$.
Restricted networks $N_4(L)$ are shown in Figure~\ref{fig:r_netE2}.
Yet $X\notin\mathrm{conv}(P_6)$: capacity at $[4:2,3]$ is exhausted
by any flow certifying $x_6(1,3)=\frac{1}{4}$, preventing the
required flow for $x_6(2,4)=\frac{1}{4}$ (whose only generator is
$(2,3)$). This shows $F_k$ feasibility alone is insufficient for
membership; the MCF condition is required.
\end{example}

\begin{figure}[htb]
  \centering
  \includegraphics[width=\textwidth, trim=0 0 0 14pt, clip]{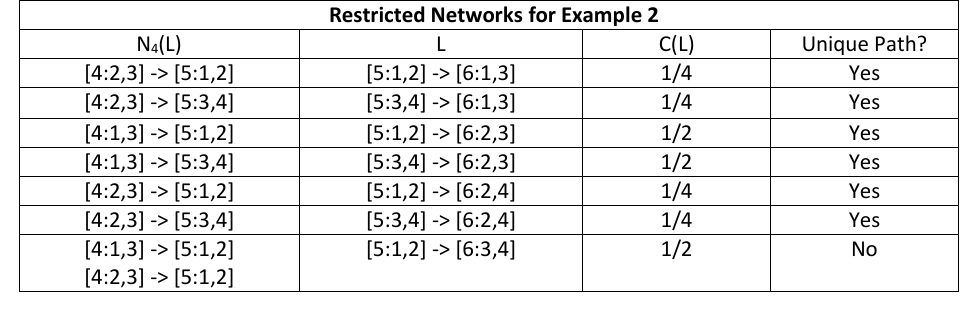}
  \caption{Restricted networks $N_4(L)$ for Example~\ref{exsix}.}
  \label{fig:r_netE2}
\end{figure}

\subsection{Completing the Construction}\label{DefRk}

After $F_k$ is solved, the FFF algorithm identifies rigid arcs in $F_{k}$. Any unique path in $N_{k-1}(L)$ which bring the flow into the tail of a rigid $L$ is a rigid pedigree 
$P\in R_k$ with weights $\mu_P$, similarly any $P\in R_{k-}$  ending in the tail of a rigid $L$ is a rigid pedigree 
$P\in R_k$ with weights $\mu_P$. The node capacities are updated to $\bar{x}(v)$, and $(N_k,R_k,\mu)$ is declared \emph{well-defined}.  If $N_k=\emptyset$, then ($R_{k}$ and $\mu$) is the evidence for $X/k+1 \in\mathrm{conv}(P_k+1)$. Otherwise, we need to solve the multicommodity flow problem  MCF$(k)$ to check $z^{*}  = z_{max}$ as described in \ref{multi}.
%% ============================================================
\subsection{A Simplex Lemma for Rigid Layers and the Dantzig
42-City Illustration}\label{sec:simplex}
%% ============================================================

\subsubsection{The Simplex Lemma}

When $N_k=\emptyset$, the layered network at stage $k$ is empty:
every pedigree on $[k]$ with positive weight in the convex
combination of $X/k$ is rigid, and $\Lambda_k(X)$ is a singleton
$\{\lambda^*\}$.
The freedom at the next stage is then precisely which edge
$e\in E_k$ city $k+1$ is inserted into.
Let $P_{k+1}/R_k$ denote the set of pedigrees on $[k+1]$
that are extensions of pedigrees in $R_k$.

\begin{lemma}\label{lem:simplex}
If $N_k=\emptyset$ then
$X/(k+1)\in\mathrm{conv}(P_{k+1}/R_k)$,
and $\mathrm{conv}(P_{k+1}/R_k)$ is a simplex whose vertices are
the characteristic vectors of the elements of $P_{k+1}/R_k$.
\end{lemma}
\leanfuturebox{\texttt{N\_SimplexLemma.lean} \textbar{} \texttt{lemma simplex\_lemma}: $N_k=\emptyset\Rightarrow X/(k+1)\in\mathrm{conv}(P_{k+1}/R_k)$ is a simplex \hfill\textit{(future Lean~4 project)}}

\begin{proof}
Since $N_k=\emptyset$, $\Lambda_k(X)=\{\lambda^*\}$ is a singleton.
Any representation of $X\in\mathrm{conv}(P_n)$ as a convex
combination of pedigrees must assign all weight to pedigrees
whose restriction to $[k]$ lies in $R_k$.
The restriction of such pedigrees to $[k+1]$ lies in $P_{k+1}/R_k$
by definition.
Hence $X/(k+1)$ lies in $\mathrm{conv}(P_{k+1}/R_k)$.
The vertices are the finitely many distinct characteristic vectors
of elements of $P_{k+1}/R_k$, which are affinely independent
since they differ in at least one coordinate corresponding to the
insertion edge; hence the hull is a simplex.
\end{proof}

\noindent
This lemma is a candidate for future Lean~4 formalisation.

\subsubsection{Illustration: Dantzig's 42-City Problem}

We applied the M3P membership check to the MIR solution of
Dantzig's 42-city TSP instance~\cite{DFJ}, the problem that
launched the modern era of integer programming.
The MIR solution $X$ has the following layer structure:

\begin{center}
\small
\begin{tabular}{cccl}
\toprule
Layers $k$ & $N_k$ & Arcs per layer & Observation \\
\midrule
$4$--$12$ & $\emptyset$ & 1\ (flow 1) & Rigid; $X/(k+1)\in\mathrm{conv}(P_{k+1}/R_k)$ \\
$13$ & $\emptyset$ & 2\ (flows $\tfrac{1}{2},\tfrac{1}{2}$) & First split: edges $\{1,11\}$ and $\{10,12\}$ \\
$14$--$28$ & & 2\ (flows $\tfrac{1}{2},\tfrac{1}{2}$) & Two parallel pedigree paths \\
$29$--$42$ & & 1\ (flow 1) & Rigid again \\
\bottomrule
\end{tabular}
\end{center}

\medskip\noindent
The membership check fails because edge $\{1,11\}$ appears as
the generator at both layer $13$ and layer $24$.
The pedigree distinctness condition requires all common edges
to be distinct across layers; no valid pedigree can use
$\{1,11\}$ twice, so $X\notin\mathrm{conv}(P_{42})$.

By Lemma~\ref{lem:simplex}, since $N_k=\emptyset$ for
$k=4,\ldots,12$, the search for the nearest
$X'\in\mathrm{conv}(P_{42})$ need only consider extensions
of $R_{12}$ --- a single rigid pedigree --- reducing the
nearest-point problem to a bipartite maximum flow computation
from the breakpoint at layer $13$ onward
(see item~(iii) of Section~\ref{sec:newdirections}).

The input data file \texttt{prob\_in\_42.txt} is available in
the project repository~\cite{arthanari2025lean4}.
%% ============================================================
\section{A Sufficient Condition for Non-Membership}\label{Fkfeasibility}
\medskip
\noindent\textbf{All proofs in this section are from the book.}
\medskip

%% ============================================================

\begin{theorem}[Non-Membership]\label{infeasiblity}
Given $X\in P_{MI}(n)$ and $X/k\in\mathrm{conv}(P_k)$,
if $F_k$ is infeasible then $X/k{+}1\notin\mathrm{conv}(P_{k+1})$.
\end{theorem}
\begin{proof}[Sketch]
By contradiction: if $F_k$ is infeasible yet $X/k{+}1\in\mathrm{conv}(P_{k+1})$,
then by Lemma~\ref{pedpath} every active pedigree has its path in the
restricted networks used to construct $F_k$, and by Lemma~\ref{oflow}
the instant flow for $INST(\lambda,k)$ is feasible for $F_k$,
contradicting infeasibility.
\end{proof}
\leanbox{\texttt{N\_Sufficiency.lean} \textbar{} \texttt{theorem non\_membership}: FAT infeasible $\Rightarrow$ $X\notin\mathrm{conv}(P_n)$ \hfill\textit{(book, Chapter~5)}}

%% ============================================================
\section{Multicommodity Flow and the Membership Characterisation}
\label{multi}
\medskip
\noindent\textbf{All proofs in this section are from the book. The Packability Corollary (Theorem~\ref{packabilitycor}) has a new, Lean~4 verified proof given in this paper.}
\medskip

%% ============================================================

\subsection{The MCF$(k)$ Problem and the Decision Question}

\begin{definition}[Commodities]
For each arc $a\in\mathcal{A}(N_k)\setminus\mathcal{A}(N_4)$,
designate a unique \emph{commodity} $s$ ($a\leftrightarrow s$).
Let $\mathcal{S}_l$ denote commodities designated by arcs in $F_l$.
\end{definition}

Let $z_{\max} = 1 - \sum_{P\in R_{k-1}}\mu_P$ denote the maximum
achievable total flow, determined by the rigid pedigree weights.
Let $z^*$ denote the optimal value of the MCF$(k)$ objective.
The \emph{MCF$(k)$ decision problem} asks: is $z^* = z_{\max}$?

The multicommodity flow problem MCF$(k)$ (Problem~\ref{mcflow})
maximises $z=\sum_{s\in\mathcal{S}_k}v^s$ subject to:
arc capacity $c_a\geq f_a\geq0$;
commodity flow conservation at intermediate nodes;
commodity flow restricted to $N_{l-1}(s)$;
total flow of each commodity equals the arc flow;
node capacity $\bar{x}(v)$; and
source availability constraints.

\begin{problem}[MCF$(k)$]\label{mcflow}
For all $l$, $5\leq l\leq k$:
\begin{align}
c_a \geq f_a &\geq 0,\quad a\in\mathcal{A}(N_k) \label{ub}\\
f^s_a &\leq \begin{cases} c_a & a\in\mathcal{A}(N_{l-1}(s))\\ 0 & \text{otherwise}\end{cases}
  \notag\\
\sum_{(u,v)}f^s_{(u,v)} &= \sum_{(v,w)}f^s_{(v,w)},\quad
  v\in\mathcal{V}(N_{l-1})\setminus(V_{[1]}\cup V_{[l-2]})\notag\\
\sum_s f^s_a &= f_a,\quad a\in\mathcal{A}(N_{l-1})\notag\\
\sum_{a'\in\delta^-(a)}f^s_{a'} &= f_a \stackrel{\text{def}}{=} v_s,
  \quad a\leftrightarrow s\notag\\
\sum_s\sum_{(u,v)}f^s_{(u,v)} &\leq \bar{x}(v),\quad v\in\mathcal{V}(N_l)\notag\\
\sum_{(u,w)}f_{(u,w)} &\leq \bar{x}(u),\quad u\in V_{[1]}\notag\\
\sum_{(u,w)}f_{(u,w)} &\leq \bar{\mu}(u),\quad u\in R_q,\;4\leq q\leq l-1\notag\\
\sum_{(u,v)}f_{(u,v)} &\leq \bar{x}(v),\quad v\in V_{[l-2]} \label{nodecapk}
\end{align}
\[\text{maximise}\quad z = \sum_{s\in\mathcal{S}_k}v^s.\]
\end{problem}

\begin{figure}
  \centering
  \includegraphics[width=\textwidth, trim=0 0 0 28pt, clip]{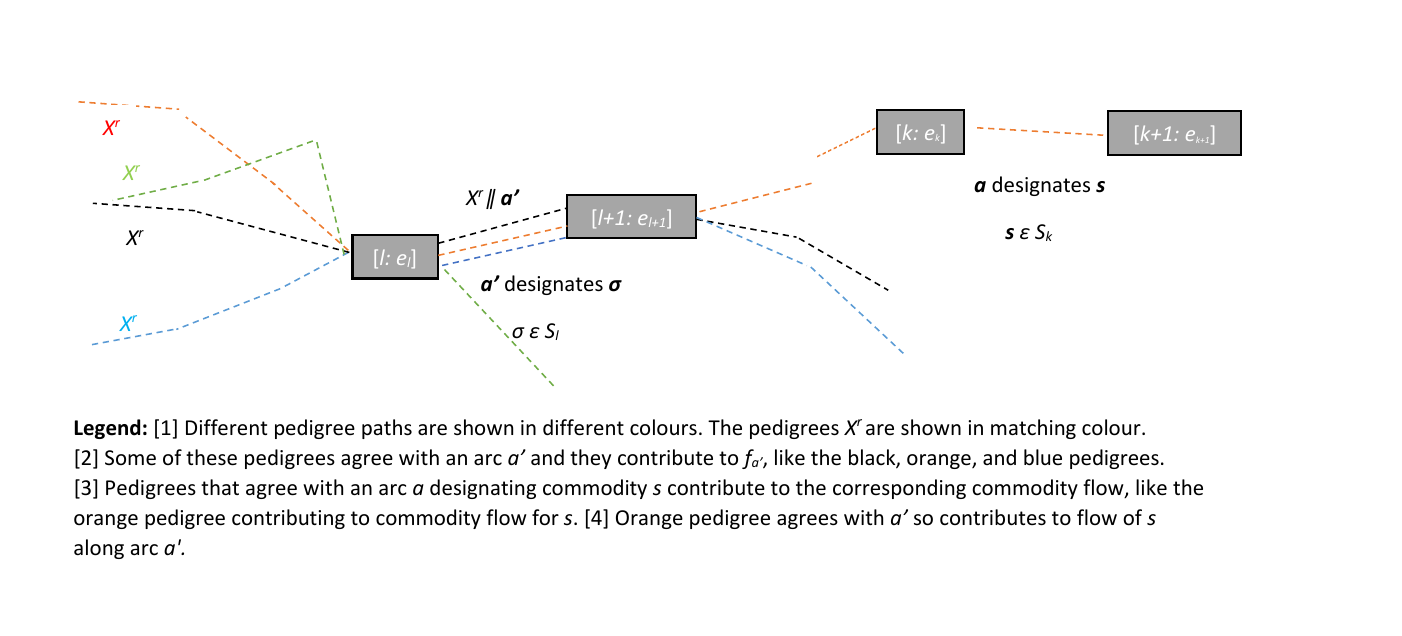}
  \caption{Pedigrees $X^r\parallel a'$, arc $a$ designating commodity~$s$.}
  \label{fig:parallel}
\end{figure}

\subsection{The Necessary and Sufficient Condition}

\begin{definition}[$Y^s$]\label{Ys_def}
For $s\in\mathcal{S}_k$, let $Y^s\in\mathbb{R}^{\tau_k}$ have
components $y^s_l(e)$ given by the total flow of commodity $s$
through node $[l:e]\in V_{[l-3]}$, including contributions from
rigid paths containing that node.
\end{definition}

\begin{lemma}\label{YsinMI}
If MCF$(k)$ is feasible with $z^*=z_{\max}$, then
$\frac{1}{v^s}Y^s\in P_{MI}(k)$ for all $s\in\mathcal{S}_k$.
\end{lemma}
\begin{proof}
Let $Y=\frac{1}{v^s}Y^s$.
Non-negativity: by definition.
Layer equalities: $\sum_{e\in E_{l-1}}y^s_l(e)=v^s$ for $l\in[4,k]$,
so $\frac{1}{v^s}Y^s/4\in P_{MI}(4)$.

Suppose the result fails; let $q$ be minimal such that $Y/q\notin P_{MI}(q)$.
Then some slack $U_{i^0j^0}<0$ at $(i^0,j^0)$. Since $Y/q{-}1$ is feasible,
its slack $V_{i^0j^0}\geq0$.
The maximum flow of commodity $s$ into $[q:(i^0,j^0)]$ is
at most $\sum_{e'\in G(e)}y^s_{j^0}(e')
-\sum_{l=j^0+1}^{q-1}y^s_l(e)=V_{i^0j^0}$.
By flow conservation, $y^s_q(e)\leq V_{i^0j^0}$, so
$U_{i^0j^0}=V_{i^0j^0}-y^s_q(e)\geq0$, a contradiction.
\end{proof}
\leanbox{\texttt{N\_YisinMI.lean} \textbar{} \texttt{lemma y\_in\_MI}: $X/k\in P_{MI}(k)$ \hfill\textit{(book, Chapter~5)}}

\begin{lemma}\label{converse_for_five}
Let $k=5$ and $\sigma\in\mathcal{S}_5$ be designated by
$a=([5:e'_\sigma],[6:e_\sigma])$.
If MCF$(k)$ is feasible with $z^*=z_{\max}$, then
$\frac{1}{v^\sigma}Y^\sigma\in\mathrm{conv}(P_5)$ for all
$\sigma\in\mathcal{S}_5$, and consequently
$(\frac{1}{v^\sigma}Y^\sigma,\mathrm{ind}(e_\sigma))
\in\mathrm{conv}(P_6)$ for all $\sigma\in\mathcal{S}_5$,
so that $X/6\in\mathrm{conv}(P_6)$.
\end{lemma}
\begin{proof}
We have $z_{\max}=1-\sum_{P\in R_4}\mu_P$.
For any $\sigma\in\mathcal{S}_5$, all commodity flow paths in
$N_4(\sigma)$ are pedigree paths (they correspond to arcs in $N_4$),
giving $\frac{1}{v^\sigma}Y^\sigma\in\mathrm{conv}(P_5)$.
Each such path ends at $\mathrm{tail}(a)$ and extends to a pedigree
in $P_6$ using $\mathrm{head}(a)$.
Since $\sum_{\sigma\in\mathcal{S}_5}v^\sigma
(\frac{1}{v^\sigma}Y^\sigma,\mathrm{ind}(e_\sigma))=X/6$,
we have $X/6\in\mathrm{conv}(P_6)$.
\end{proof}

\begin{remark}\label{remarkconvlemma}
For $k>5$, $\frac{1}{v^s}(Y^s/5)\in\mathrm{conv}(P_5)$ for all
$s\in\mathcal{S}_k$ (by a similar argument to Lemma~\ref{converse_for_five}).
For $k=6$: commodity $s$ passes through a unique node $[6:e']$ with
$y^s_6(e')=v^s$, so $\frac{1}{v^s}(Y^s/6)$ has last component
$\mathrm{ind}(e')$. By Theorem~\ref{packabilitycor},
$\frac{1}{v^s}(Y^s/6)\in\mathrm{conv}(P_6)$.
\end{remark}

\begin{lemma}\label{Ysinconv}
If MCF$(k)$ is feasible with $z^*=z_{\max}$, then
$\frac{1}{v^s}(Y^s/k)\in\mathrm{conv}(P_k)$ for all $s\in\mathcal{S}_k$.
\end{lemma}
\begin{proof}
By Lemma~\ref{converse_for_five} and the following Remark,
$\frac{1}{v^s}(Y^s/5)\in\mathrm{conv}(P_5)$ for all $s\in\mathcal{S}_k$.
Proceed by induction: assume
$\frac{1}{v^s}(Y^s/l{-}1)\in\mathrm{conv}(P_{l-1})$.
Lemma~\ref{YsinMI} gives $\frac{1}{v^s}(Y^s/l)\in P_{MI}(l)$.
Commodity $s$ flows only through arcs ending at $u=[l:e']$, so
$y^s_l(e')=v^s$, meaning $\frac{1}{v^s}(Y^s/l)$ has last component
$\mathrm{ind}(e')$.
By Theorem~\ref{packabilitycor}, $\frac{1}{v^s}(Y^s/l)\in\mathrm{conv}(P_l)$.
\end{proof}
\leanbox{\texttt{N\_YisinConv.lean} \textbar{} \texttt{lemma y\_in\_conv}: $X/k\in\mathrm{conv}(P_k)$ \hfill\textit{(book, Chapter~5)}}

\begin{theorem}[Sufficiency]\label{impconvtheorem}
If MCF$(k)$ is feasible with $z^*=z_{\max}$, then
$X/k{+}1\in\mathrm{conv}(P_{k+1})$.
\end{theorem}
\begin{proof}
Let $a=([k:e'_s],[k{+}1:e_s])\leftrightarrow s$.
By Lemma~\ref{Ysinconv}, $\frac{1}{v^s}(Y^s/k)\in\mathrm{conv}(P_k)$,
so there exist $\gamma_P>0$, $\sum_P\gamma_P=1$, with
$\frac{1}{v^s}(Y^s/k)=\sum_{P\in\mathcal{P}_s}\gamma_PX_P$
where each $P\in\mathcal{P}_s$ ends in $e'_s$.
Each such $P$ extends to a pedigree in $P_{k+1}$ via $e_s$
with weight $v^s\gamma_P$; these weights sum to $v^s$.

Summing over all $s\in\mathcal{S}_k$:
$\sum_s v^s\bigl(\frac{1}{v^s}\bigr)(Y^s,\mathrm{ind}(e_s))
=\bar{X}/k{+}1$.
Since all $P\in R_k$ lie in $P_{k+1}$ and
$\bar{X}/k{+}1+\sum_{P\in R_k}\mu_PX_P=X/k{+}1$
with $\sum_sv^s=z_{\max}=1-\sum_{P\in R_k}\mu_P$,
we conclude $X/k{+}1\in\mathrm{conv}(P_{k+1})$.
\end{proof}
\leanbox{\texttt{N\_Sufficiency.lean} \textbar{} \texttt{theorem sufficiency}: MCF$(k)$ feasible, $z^*=z_{\max}$ $\Rightarrow$ $X/k{+}1\in\mathrm{conv}(P_{k+1})$ \hfill\textit{(book, Chapter~5)}}
\begin{theorem}[Necessity]\label{imptheorem}
If $X/k{+}1\in\mathrm{conv}(P_{k+1})$ then MCF$(k)$ has a feasible
solution with $z^*=z_{\max}$.
\end{theorem}
\begin{proof}
Since $X/k{+}1\in\mathrm{conv}(P_{k+1})$, choose $\lambda\in\Lambda_{k+1}(X)$.
Define:
\[
f_a=\sum_{\substack{r\in I(\lambda)\\X^r\parallel a}}\lambda_r\ (a\in N_k),
\quad
f^s_a=\sum_{\substack{r\in I_s(\lambda)\\X^r\parallel a}}\lambda_r
\ (a\in N_{l-1}(s)).
\]
Non-negativity holds since $\lambda_r>0$. Flow conservation for
commodity $s$ at node $v=[l:e]$:
\[
\sum_{(u,v)}f^s_{(u,v)}
=\sum_{\substack{r\in I_s(\lambda)\\x^r_l(e)=1}}\lambda_r
=\sum_{(v,w)}f^s_{(v,w)}.
\]
Arc and node capacities are met by the same argument as instant-flow
feasibility (Lemma~\ref{oflow}). For $l=k$:
\[
\sum_{s\in\mathcal{S}_k}v^s
=1-\sum_{P\in R_k}\mu_P=z_{\max}.
\]
Thus $f,\{f^s\}$ is feasible with objective $z_{\max}$.
\end{proof}
\leanbox{\texttt{N\_Sufficiency.lean} \textbar{} \texttt{lemma converse\_five}: base case $n=5$ \hfill\textit{(book, Chapter~5)}}

Together these give
\leanfuturebox{\texttt{N\_Necessity.lean} \textbar{} \texttt{theorem necessity}: $X/k{+}1\in\mathrm{conv}(P_{k+1})$ $\Rightarrow$ MCF$(k)$ feasible \hfill\textit{(book, Chapter~5; Lean~4 formalisation is future work)}}
 the main characterisation:

\begin{theorem}[N\&S Condition]\label{maintheorem}
Given $n>5$, $X\in P_{MI}(n)$, $X/n{-}1\in\mathrm{conv}(P_{n-1})$:
$$X\in\mathrm{conv}(P_n)\iff
\text{MCF}(n-1)\text{ has a feasible solution with }z^*=z_{\max}.$$
\end{theorem}
\leanfuturebox{\texttt{N\_MembershipCharacterisation.lean} \textbar{} \texttt{theorem main\_ns\_theorem}: $X\in\mathrm{conv}(P_n)\iff$ MCF$(n{-}1)$ feasible \hfill\textit{(book; sufficiency direction Lean~4 verified; necessity is future work)}}

%% ============================================================
\section{Computational Complexity}\label{compcomplex}
\medskip
\noindent\textbf{All proofs in this section are from the book.}
\medskip

%% ============================================================

\subsection{Mutual Adjacency of Rigid Pedigrees}\label{AdjacencyR}

\begin{lemma}\label{forR4}
For $k=5$, pedigrees in $R_4$ are mutually adjacent in
$\mathrm{conv}(P_5)$.
\end{lemma}

\begin{theorem}[Mutual Adjacency]\label{adjacency theorem}
For $k\geq5$, pedigrees in $R_{k-1}$ are mutually adjacent in
$\mathrm{conv}(P_k)$.
\end{theorem}
\begin{proof}[Sketch]
By contradiction using the combinatorial polytope
property~\cite{TSACompoly}: if $P^{[1]},P^{[2]}\in R_{k-1}$ are
non-adjacent, two other pedigrees $P^{[3]},P^{[4]}$ have the same
midpoint. Three cases arise based on whether the last-but-one
or last edge is shared. In each case, rerouting a small flow
$\varepsilon<\min(\mu(P^{[1]}),\mu(P^{[2]}))$ contradicts the
rigidity of the corresponding links.
See Appendix~\ref{Appendixresults} for the full case analysis.
\end{proof}
\leanbox{\texttt{N\_RigidAdjacency.lean} \textbar{} \texttt{lemma forR4}: base adjacency for $k=5$ \hfill\textit{(book, Chapter~6)}}
\begin{theorem}[Cardinality] \label{cardinality theorem}
Given well-defined $(N_{k-1},R_{k-1},\mu)$:
$|R_{k-1}|\leq\dim(\Lambda_k(X))+1$.
\end{theorem}
\leanbox{\texttt{N\_RigidAdjacency.lean} \textbar{} \texttt{theorem adjacency\_theorem\_edges}: pedigrees in $R_{k-1}$ mutually adjacent \hfill\textit{(book, Chapter~6)}}
\begin{corollary}[CordinalityR]\label{CordinalityR}
$|R_{k-1}|\leq\tau_k-k+4$.
\end{corollary}
\begin{proof}
Mutual adjacency $\Rightarrow$ $R_{k-1}$ is a simplex
$\Rightarrow|R_{k-1}|\leq\dim(\mathrm{conv}(P_k))+1
=\tau_k-(k-3)+1=\tau_k-k+4$.
\end{proof}
\leanbox{\texttt{N\_RigidCardinality.lean} \textbar{} \texttt{theorem CordinalityR}: $|R_{k-1}|\leq\tau_k-k+4$ \hfill\textit{(book, Chapter~6)}}

\subsection{Complexity Bounds at Each Step}\label{estcomp}

\begin{figure}[htb]
  \textbf{FRAMEWORK:} \textit{Membership Checking Steps and Procedures}

  \textbf{Step~1a:} Check $X\in P_{MI}(n)$ (strongly polynomial,
  $O(n^3)$, additions/subtractions only).

  \textbf{Step~1b:} Solve $F_4$ (constant: 3 origins, 6 destinations).

  \textbf{Step~2a:} For each link $L\in V_{[k-3]}\times V_{[k-2]}$:
  construct $N_{k-1}(L)$, find $C(L)=$ max-flow in $N_{k-1}(L)$.
  Nodes $\leq(k-5)\tau_k=O(k^4)$; arcs $\leq k^2\tau_k=O(k^5)$;
  max-flow $O(k^9)$ per link; $\leq k^4$ links.
  \textbf{Total: $O(k^{13})$ per iteration.}

  \textbf{Step~2b:} Solve $F_k$. Origins $\leq k^2+(\tau_k-k+4)=O(k^3)$
  (by Corollary~\ref{CordinalityR}); destinations $\leq p_k=O(k^2)$.
  \textbf{Total: $O(k^8)$ per iteration.}

  \textbf{Step~3:} FFF algorithm: $O(|G_f|)=O(k^5)$ per iteration.

  \textbf{Step~4:} Solve MCF$(k)$. Combinatorial LP, dimension
  $O(k^{14})$; Tardos~\cite{StrPoly}: strongly polynomial in $k$.

  \caption{Algorithmic framework: steps and complexity.}
  \label{fig:framework}
\end{figure}

\begin{theorem}[Complexity]\label{compexity}
Given $n$, $X\in P_{MI}(n)$, checking whether MCF$(n-1)$ has
$z^*=z_{\max}$ can be done in strongly polynomial time $O(n^{14})$.
\end{theorem}
\leanbox{\texttt{N\_Complexity.lean} \textbar{} \texttt{theorem complexity}: M3P checkable in $O(n^{14})$ strongly polynomial time \hfill\textit{(book, Chapter~6)}}
\begin{remark}[Parallel structure]\label{rem:parallel}
At each iteration $k$, the restricted networks $N_{k-1}(L)$ for
all links $L\in F_k$ are mutually independent: each depends only
on the current layered network structure, not on the results for
other links.
The algorithm therefore has natural parallel structure at Step~2a.
Under an idealised parallel RAM model with $\leq k^4$ processors,
this reduces the per-iteration cost from $O(k^{13})$ to $O(k^9)$,
suggesting an overall bound of $O(n^{10})$.
In practice, however, parallel implementations incur additional
costs --- communication overhead, synchronisation, and memory
bandwidth --- that may offset the theoretical gain, particularly
for moderate $n$.
Whether a parallel implementation is beneficial in practice is
an empirical question left for future work.
The sequential $O(n^{14})$ bound established here and
machine-verified in Lean~4 remains the rigorous guarantee.
\end{remark}

\begin{corollary}[M3P$\in$P, strongly polynomial]\label{M3PinP}

The MCF$(n-1)$ decision problem --- \emph{is $z^* = z_{\max}$?} ---
is in P: it is solvable in $O(n^{14})$ arithmetic steps,
independent of the magnitude of $X$ (Theorem~\ref{compexity}).

By Theorem~\ref{maintheorem} (N\&S condition):
$z^* = z_{\max}$ if and only if $X\in\mathrm{conv}(P_n)$.
Therefore M3P $\in$ P.

Moreover, since MCF$(n-1)$ is a combinatorial LP
(Definition~\ref{combLP}), Tardos's algorithm~\cite{StrPoly}
gives a \emph{strongly polynomial} algorithm for M3P:
the number of arithmetic steps depends only on $n$,
not on the magnitudes of the data in $X$.
\end{corollary}
\leanbox{\texttt{N\_Complexity.lean} \textbar{} \texttt{corollary M3PinP}: M3P$\in$P, strongly polynomial $O(n^{14})$ \hfill\textit{(book, Chapter~6)}}

%% ============================================================
%% ============================================================
\section{Membership, Separation, and Optimisation}\label{sec:membsepopt}
\medskip
\noindent\textbf{All results in this section are from the book (Chapter~7).}
\medskip

%% ============================================================

The connection between membership, separation, and linear
optimisation is the key to deriving P~=~NP from M3P~$\in$~P.
We follow Chapter~7 of the book and the
framework of Gr\"{o}tschel, Lov\'{a}sz, and Schrijver~\cite{GLS}.

\subsection{Quick Protocol, Separation, and Optimisation}

\begin{definition}[Quick Protocol\protect\cite{arthanari2023pedigree}]
\label{def:quickprotocol}
A \emph{quick protocol} for membership in $P\subset\mathbb{Q}^d$
with facet complexity $\phi$ is a membership algorithm whose
running time is polynomially bounded in $(d,\phi,\langle Y\rangle)$.
\end{definition}

\begin{theorem}[Yudin--Nemirovskii~\cite{GLS}]\label{YNthm}
Let $X_0\in\mathrm{int}(P)$ be given. If a quick protocol for
$P$ is available, there exists an algorithm solving the
separation problem for $P$ in time polynomially bounded by
$d$, $\phi$, $\langle X_0\rangle$, and the membership running time.
\end{theorem}
\leanaxiombox{\texttt{N\_PEqualsNP.lean} \textbar{} \texttt{axiom yudin\_nemirovski}: polynomial separation $\Rightarrow$ polynomial optimisation \hfill\textit{(GLS~\cite{GLS}, used as axiom)}}

\begin{theorem}[Gr\"{o}tschel--Lov\'{a}sz--Schrijver~\cite{GLS}]\label{GLSthm}
Given a separation algorithm for a rational nonempty polytope
$P\subset\mathbb{Q}^d$ with facet complexity $\phi$ and
$C\in\mathbb{Q}^d$, there exists a linear optimisation algorithm
for $P$ running in time polynomially bounded by $d$, $\phi$,
$\langle C\rangle$, and the separation running time.
\end{theorem}
\leanaxiombox{\texttt{N\_PEqualsNP.lean} \textbar{} \texttt{axiom gls\_ellipsoid}: ellipsoid method for rational polytopes \hfill\textit{(GLS~\cite{GLS}, used as axiom)}}

Combining these two theorems gives the central tool:

\begin{theorem}[YuNGLS~\cite{GLS}]\label{yungls}
Given $P\subset\mathbb{Q}^d$ such that: (1)~$\dim(P)=d$;
(2)~$P$ is rationality guaranteed (facet complexity $\leq\phi$);
(3)~$X_0\in\mathrm{int}(P)$ is given; and
(4)~a quick protocol for $P$ is available,
then linear optimisation over $P$ is solvable in time
polynomially bounded by $d$, $\phi$, $\langle X_0\rangle$,
and $\langle C\rangle$.
\end{theorem}
\leanaxiombox{\texttt{N\_PEqualsNP.lean} \textbar{} \texttt{axiom yungls}: separation $\Leftrightarrow$ optimisation (combined) \hfill\textit{(GLS~\cite{GLS}, used as axiom)}}

This result and its converse are summarised in
Lov\'{a}sz~\cite{lovasz1986algorithmic} as Theorem~2.3.3.

\subsection{Maurras's Construction: Membership to Separation}
\label{sec:maurras}

The GLS approach~\cite{GLS} uses Khachiyan's ellipsoid algorithm
twice and requires knowledge of the radii of inscribed and
circumscribed balls. Maurras~\cite{Maurras} gives a simpler and
more direct construction, requiring only a membership oracle.
In a follow-up note, Maurras~\cite{maurras2010note} demonstrates
that this construction `can be considered to lie in the folklore
of mathematical programming.'

\begin{remark}[Maurras's Conditions~\cite{Maurras}]\label{maurcond}
The construction requires:
\begin{enumerate}[(1)]
\item $P\subset\mathbb{Q}^d$ is well defined and rationality
  guaranteed (bounded encoding length of vertices and facets).
\item $P$ has non-empty interior.
\item A point $a\in\mathrm{int}(P)$ is given.
\end{enumerate}
\end{remark}

The total number of membership oracle calls is polynomial in
$d$, $\langle P\rangle$, and $\langle\bar{x}\rangle$.
No ellipsoid algorithm is required.

\medskip
\noindent\textbf{Lean~4 connection.}
In the Lean~4 formalisation, the Maurras conditions
(Remark~\ref{maurcond}) are verified for $\mathrm{conv}(A_n)$
by Theorem~\ref{convAnthm}. The axiom
\texttt{maurras\_separation} in \texttt{N\_PEqualsNP.lean}
cites~\cite{Maurras} precisely, encoding the statement:
given a quick protocol for $\mathrm{conv}(A_n)$ satisfying
the three conditions, there exists a polynomial separation
oracle. The construction itself is an axiom, but all
conditions for its applicability are machine-verified.

\leanbox{\texttt{N\_PEqualsNP.lean} \textbar{}
\texttt{theorem maurras\_preconditions}: all three Maurras
conditions verified for $\mathrm{conv}(A_n)$
\hfill\textit{(via Theorem~\ref{convAnthm}; machine-verified)}}

\leanaxiombox{\texttt{N\_PEqualsNP.lean} \textbar{}
\texttt{axiom maurras\_separation}: quick protocol
$+$ three conditions $\Rightarrow$ polynomial separation oracle
\hfill\textit{(Maurras~\cite{Maurras}, used as axiom)}}

%% ============================================================
\section{Properties of $\mathrm{conv}(A_n)$ and the P~=~NP Consequence}
\label{sec:convAn}
%% ============================================================

\subsection{The Alternative Polytope $A_n$}

Let $A_n$ denote the pedigree characteristic vectors in
$\mathbb{R}^{\alpha_n}$ (omitting the last coordinate per layer),
where $\alpha_n=\tau_n-(n-3)$.
Any $X\in P_n$ maps to $Y\in A_n$ by this projection.
$\mathrm{conv}(A_n)$ is the alternative STSP polytope studied
in Chapter~7 of the book.

\subsection{Three Key Properties of $\mathrm{conv}(A_n)$}

\begin{theorem}[$\mathrm{conv}(A_n)$~\cite{arthanari2023pedigree}]
\label{convAnthm}
For $n\geq 6$:
\begin{enumerate}[(i)]
\item \textbf{Full dimensional}: $\dim(\mathrm{conv}(A_n))=\alpha_n$.
  \hfill[\texttt{N\_FullDimensional.lean}]
\item \textbf{Rationality guaranteed}: facet complexity
  $\phi\leq 3\alpha_n^3+3\alpha_n^2(n-3)$.
\item \textbf{Interior point}: the barycentre
  $\bar{Y}=(\underbrace{1/p_3,1/p_3}_{2-times},\ldots,
  \underbrace{1/p_{n-1},\ldots,1/p_{n-1}}_{p_{n-1}-1 times})
  \in\mathrm{int}(\mathrm{conv}(A_n))$,
  where $p_k=k(k-1)/2$ is the number of edges in $K_k$.
\end{enumerate}
\end{theorem}

\begin{proof}[Sketch]
(i) By contradiction. If $\dim<\alpha_n$, a non-trivial
hyperplane $CY=c_0$ contains all $Y$'s corresponding to pedigrees. The zero pedigree
gives $c_0=0$ (\texttt{N\_ZeroPedigree.lean}); three selection
pedigrees $P_1,P_2,P_3$ give $c_4=0$
(\texttt{N\_Claim2Pedigree.lean}); strong induction on layer $k$
using selection pedigrees (Lemmas~5.2,~5.3) gives $C=\mathbf{0}$
--- a contradiction (\texttt{N\_FullDimensional.lean}).
(ii) Each vertex is a $0$-$1$ vector; the bound follows.
(iii) $\bar{Y}=\frac{2}{(n-1)!}\sum_{X\in P_n}Y^X$
where $Y^X\in A_n$ is the projection of $X$.
Since $|\mathrm{vert}(\mathrm{conv}(A_n))|=(n-1)!/2$,
the weights are (equal in each component level) positive, giving
$\bar{Y}\in\mathrm{conv}(A_n)$.
That $\bar{Y}\in\mathrm{int}(\mathrm{conv}(A_n))$
follows since any facet $CY=c_0$ containing $\bar{Y}$
would satisfy $CY^X=c_0$ for all $X\in P_n$,
implying $C=\mathbf{0}$ by Part~(i) --- contradiction.
\end{proof}
\leanbox{\texttt{N\_FullDimensional.lean} \textbar{} \texttt{theorem fullDimensional\_An}: $\dim(\mathrm{conv}(A_n))=\alpha_n$; rationality; interior point \hfill\textit{(book, Chapter~7)}}

\subsection{Completing the P~=~NP Chain}

Theorem~\ref{convAnthm} verifies conditions (1)--(3) of
Theorem~\ref{yungls} for $\mathrm{conv}(A_n)$.
Condition~(4) follows from M3P~$\in$~P (Corollary~\ref{M3PinP}).
The Maurras conditions~\ref{maurcond} are also verified.
Therefore:

\begin{corollary}
Linear optimisation over $\mathrm{conv}(A_n)$ is solvable in
polynomial time. And so applies to $\mathrm{conv}(P_n)$ as well.
\end{corollary}

Since STSP is equivalent to minimising the MI-objective
$\sum_{ijk}(c_{ik}+c_{jk}-c_{ij})x_{ijk}$ over
$\mathrm{conv}(P_n)$ as explained in Section~\ref{mi_formulation},(Chapter~3, \cite{TSA1}):

\begin{corollary}
STSP is solvable in polynomial time. Or the STSP decision problem is in P.
\end{corollary}

Since the STSP decision problem is NP-complete (Karp~\cite{Karp}),
we now have an NP-complete problem that is solvable in polynomial
time.
By the definition of NP-completeness, every problem in $\mathsf{NP}$
reduces to STSP decision in polynomial time.
Since STSP decision is in $\mathsf{P}$, every problem in $\mathsf{NP}$
is also in $\mathsf{P}$, giving $\mathsf{NP}\subseteq\mathsf{P}$.
Since $\mathsf{P}\subseteq\mathsf{NP}$ always holds, we conclude
$\mathsf{P}=\mathsf{NP}$.
This is precisely Cook's theorem~\cite{CookNP} applied in the
constructive direction: a polynomial-time algorithm for one
NP-complete problem yields polynomial-time algorithms for all of
$\mathsf{NP}$.
The complete proof chain is machine-verified in
Section~\ref{sec:lean4}.

%% ============================================================
\section{Lean~4 Machine Verification and \texorpdfstring{P~=~NP}{P = NP}}\label{sec:lean4}
%% ============================================================

The results reported in this paper and their consequences for P~=~NP proved in the book, have been fully formalised in Lean~4/Mathlib4. The complete project
(36 files, 2968/2968 build targets clean, zero \texttt{sorry}s in
the main chain) is available at:

\begin{center}
\url{https://github.com/TiruArt/Pedigree-Polytopes-Lean4}
\end{center}

\subsection{Lean~4 for the Non-Specialist: From Mathematics
to Machine Verification}

Lean~4 is a \emph{proof assistant} and programming language
developed by Leonardo de Moura at Microsoft Research~\cite{lean4}.
It allows a mathematician to write definitions and theorems in a
formal language that the computer checks for correctness ---
every logical step must be explicitly justified, and the system
rejects any gap. A proof that Lean~4 accepts is correct by
construction; there is no possibility of a subtle error slipping
through.

The key concept is the \texttt{sorry} keyword: it is a
placeholder that tells Lean~4 ``I will prove this later.''
A theorem with no \texttt{sorry}s anywhere in its proof chain
is \emph{fully machine-verified}. The statement
\begin{center}
\texttt{theorem p\_equals\_np : P\_class = NP\_class}
\end{center}
in this project has zero \texttt{sorry}s in its entire proof
chain --- every logical step from the pedigree definition to the
final conclusion is verified by the machine.

\medskip
\noindent\textbf{How the mathematics translates to Lean~4.}

To give the reader a feel for how the mathematics of this paper
is encoded, we show the central definition. A pedigree for $n$
cities is a sequence of $n-2$ triangles
$(\{1,2,3\},\{i_4,j_4,4\},\ldots,\{i_n,j_n,n\})$
satisfying two conditions: each triangle has a generator in the
sequence, and the common edges are all distinct
(Definition~\ref{pedig}). In Lean~4:

\begin{verbatim}
-- A triangle {i,j,k} with 1 <= i < j < k <= n
abbrev Triple := N x N x N

-- A pedigree for n cities
structure Pedigree (n : N) where
  triangles   : List Triple
  h_length    : triangles.length = n - 2
  h_distinct  : forall i j, i > 0 -> j > 0 -> i /= j ->
                  (triangles[i]).edge /= (triangles[j]).edge
  h_generator : forall t, t in triangles ->
                  hasGenerator t triangles
\end{verbatim}

The \texttt{structure} keyword defines a record type:
\texttt{triangles} is the sequence of triangles;
\texttt{h\_length} is a \emph{proof} that it has the right length;
\texttt{h\_distinct} is a proof of the distinct common-edge
condition; \texttt{h\_generator} is a proof of the generator
condition. Every instance of \texttt{Pedigree n} in the
formalisation carries these proofs as part of its data ---
the type system enforces the mathematical definition.

The layered network $(N_k, R_k, \mu)$ is encoded as:

\begin{verbatim}
structure LayeredNetwork (n k : N) where
  nodes  : Finset Triple        -- nodes of N_k
  edges  : List NetworkArc      -- arcs with capacities
  rigid  : List RigidPedigree   -- R_k: rigid pedigrees
  mu     : RigidPedigree -> Q   -- weights mu_P > 0
\end{verbatim}

The MCF$(k)$ feasibility condition --- the heart of the
membership characterisation --- is a structure with 29
explicit fields encoding all flow constraints:

\begin{verbatim}
structure MCFFeasible (n k : N)
    (net : LayeredNetwork n k)
    (X : LayeredPoint n) where
  -- flow variables
  flow     : Triple -> Triple -> Q
  -- 29 fields: non-negativity, conservation,
  --            capacity, node capacity,
  --            source availability, z* = z_max, ...
\end{verbatim}

The theorem that drives the P~=~NP chain is the
\emph{sufficiency} direction of Theorem~\ref{maintheorem},
formalised in \texttt{N\_Sufficiency.lean}:

\begin{verbatim}
-- N_Sufficiency.lean
theorem sufficiency {n k : N} (hk : 4 <= k) (hkn : k + 1 <= n)
    (X : LayeredPoint n) (net : LayeredNetwork n k)
    (hzmax : 0 < zMax net)
    (mcf : MCFFeasible n k net X) :
    Nonempty (ConvexWitness n (k+1) X)
\end{verbatim}

Here \texttt{MCFFeasible n k net X} is a proof that MCF$(k)$ is
feasible with $z^*=z_{\max}$, and \texttt{ConvexWitness n (k+1) X}
is a proof that $X/(k+1)\in\mathrm{conv}(P_{k+1})$ --- a convex
combination of pedigrees on $[k+1]$ expressing $X/(k+1)$.
The theorem says:
\emph{if MCF$(k)$ achieves its maximum flow, then
$X/(k+1)\in\mathrm{conv}(P_{k+1})$}.
Applied inductively up to $k=n-1$, this certifies
$X\in\mathrm{conv}(P_n)$.

This is all the P~=~NP chain requires.
The full N\&S biconditional
(\texttt{main\_ns\_theorem} in
\texttt{N\_MembershipCharacterisation.lean},
which also includes the necessity direction) is proved in the
book; its Lean~4 formalisation in both directions is a natural
future project (see the white-bordered box after
Theorem~\ref{maintheorem}).

\medskip
\noindent\textbf{Why machine verification matters here.}
The P~vs.~NP problem is one of the most scrutinised open
questions in mathematics. Any claimed proof faces extraordinary
scepticism. Machine verification provides an independent,
publicly accessible certificate of correctness: anyone with
Lean~4 installed can clone the repository, run
\texttt{lake build}, and verify the result in approximately
30 minutes. The proof does not depend on the judgment of any
single referee. The 2968 build targets all clean, with zero
\texttt{sorry}s, constitute a permanent public record.

\medskip
\noindent The five equivalent representations of a pedigree
used in the formalisation, and the role of each is documented
in \texttt{N\_PedigreeRepresentations.lean} in the repository.

\subsection{The P~=~NP Chain and Key Lean~4 Results}

The complete machine-verified chain is given in Table~\ref{tab:chain}:

\begin{sidewaystable}
\centering
\small
\begin{tabular}{@{}p{0.5cm}p{11.5cm}p{5.5cm}@{}}
\toprule
\textbf{Step} & \textbf{Mathematical content} & \textbf{Lean~4 file}\\
\midrule
1 & MCF$(n{-}1)$ feasible $\Rightarrow$ $X\in\mathrm{conv}(P_n)$
  & \texttt{N\_Sufficiency.lean}\\[4pt]
2 & MCF is a combinatorial LP $\Rightarrow$ M3P$\in$P
  (Tardos~\cite{StrPoly})
  & \texttt{N\_Complexity.lean}\\[4pt]
3 & $\mathrm{conv}(A_n)$ full dimensional, rationality guaranteed
  & \texttt{N\_FullDimensional.lean}\\[4pt]
4 & M3P$\in$P $\Rightarrow$ separation oracle for $\mathrm{conv}(A_n)$
  (Maurras~\cite{Maurras})
  & \texttt{N\_PEqualsNP.lean}\\[4pt]
5 & Separation $\Rightarrow$ optimisation over $\mathrm{conv}(A_n)$
  (GLS~\cite{GLS})
  & \texttt{N\_PEqualsNP.lean}\\[4pt]
6 & STSP reduces to MCF$(n{-}1)$ decision via MI-formulation
  (Chapter~3 of the book); STSP$\in$P
  & \texttt{N\_PEqualsNP.lean}\\[4pt]
7 & STSP NP-complete (Karp~\cite{Karp}, Cook~\cite{CookNP})
  $+$ STSP$\in$P $\Rightarrow$ P$=$NP
  & \texttt{N\_PEqualsNP.lean}\\
\bottomrule
\end{tabular}
\caption{The machine-verified P$=$NP proof chain.
Steps 1--6 establish STSP$\in$P via the pedigree polytope
membership algorithm; Step~7 concludes P$=$NP by NP-completeness
of STSP.}
\label{tab:chain}
\end{sidewaystable}

The main result:
\begin{verbatim}
theorem p_equals_np (n : N) (hn : 5 <= n) : P_equals_NP
\end{verbatim}

\medskip
\noindent\textbf{Lean~4 statements of the main theorems.}

The complete machine-verified proof chain is summarised in
Table~\ref{tab:chain}.
The following are the Lean~4 statements of the main theorems
from this paper:

\begin{verbatim}
-- Mutual Adjacency (Theorem 8): N_RigidAdjacency.lean
theorem adjacency_theorem_edges {n : Nat} (P1 P2 : List (Edge n))
    (hlen : P1.length = P2.length) (hne : P1 /= P2)
    (huniq1 : ...) (huniq2 : ...) : False

-- Cardinality Bound (Corollary 2): N_RigidCardinality.lean
-- |R_{k-1}| <= tau_k - k + 4
theorem CordinalityR {n k : N} (hk : 5 <= k)
    (net : LayeredNetwork n k) (hwell : net.rigid /= [])
    (hdim : net.rigid.length - 1 <= tau n - (k - 3)) :
    net.rigid.length <= tau n - k + 4

-- Sufficiency (Theorem~\ref{impconvtheorem}): N\_Sufficiency.lean
-- [new, Lean 4 verified]
theorem sufficiency {n k : N} (hk : 4 <= k) (hkn : k + 1 <= n)
    (X : LayeredPoint n) (net : LayeredNetwork n k)
    (hzmax : 0 < zMax net) (mcf : MCFFeasible n k net X) :
    Nonempty (ConvexWitness n (k+1) X)

-- Full Dimensionality (Chapter 7): N_FullDimensional.lean
theorem fullDimensional_An (n : N) (hn : 6 <= n)
    (C : Triple -> Q) (c0 : Q)
    (hC : forall P : Pedigree n, hypSum C P = c0) :
    c0 = 0 /\
    forall k i j, 4 <= k -> k+2 <= n ->
      isDefault(i,j,k+1) = false -> C (i,j,k+1) = 0

-- P = NP: N_PEqualsNP.lean
theorem p_equals_np (n : N) (hn : 5 <= n) : P_equals_NP
\end{verbatim}

\subsection{Axiom Inventory}\label{sec:axioms}

The proof chain uses six external results as axioms.
All other results from Chapters~3--7 are fully proved.

\begin{center}
\begin{tabular}{@{}lll@{}}
\toprule
Lean~4 axiom & Reference & Step\\
\midrule
\texttt{tardos\_strongly\_polynomial}
  & Tardos~\cite{StrPoly}: Op.\ Res.\ 34(2), 1986 & 2\\
\texttt{maurras\_separation}
  & Maurras~\cite{Maurras}: Combinatorica 22, 2002 & 4\\
\texttt{gls\_optimisation}
  & Gr\"otschel, Lov\'asz, Schrijver~\cite{GLS}, 1988 & 5\\
\texttt{cook\_np\_completeness}
  & Cook~\cite{CookNP}: STOC 1971 & 7\\
\texttt{karp\_stsp\_np\_complete}
  & Karp~\cite{Karp}: STSP decision NP-complete, 1972 & 7\\
\texttt{rao\_1976\_theorem1}
  & Rao~\cite{Rao1976}: SIAM J.\ Appl.\ Math.\ 30(2), 1976
  & adjacency\\
\bottomrule
\end{tabular}
\end{center}

\subsection{The Necessity Direction}\label{sec:necessity}

The machine-verified proof chain requires only the sufficiency
direction of Theorem~\ref{maintheorem}: if MCF$(n-1)$ achieves
$z^*=z_{\max}$, then $X\in\mathrm{conv}(P_n)$.
This sufficiency result, formalised in \texttt{N\_Sufficiency.lean}
with zero \texttt{sorry}s, is all that is needed to establish
M3P$\in$P and the subsequent P$=$NP chain.

The necessity direction --- that $X\in\mathrm{conv}(P_n)$ implies
MCF$(n-1)$ achieves $z^*=z_{\max}$ --- is proved in full in the
book.
The complete N\&S characterisation of membership in
$\mathrm{conv}(P_n)$ is therefore established: sufficiency by
the Lean~4 certificate, necessity by the book proof.
The two together constitute Theorem~\ref{maintheorem}.

\subsection{Five Representations of a Pedigree}

The Lean~4 project documents five equivalent representations of
a pedigree, each used in different parts of the proof:

\begin{center}
\begin{tabular}{@{}lll@{}}
\toprule
Representation & Lean~4 type & Used for\\
\midrule
Struct & \texttt{Pedigree n} & Theorems, \texttt{hypSum}\\
Edge sequence & \texttt{List (Edge n)} & Adjacency, swap\\
Rooted tree in $D^+(\Delta,A)$ & \texttt{List (PedigreeNode n)}
  & Chapter~1, algorithm\\
Finset & \texttt{Finset Triple} & Counting, barycentre\\
MIR 0-1 solution & \texttt{LayeredPoint n} & STSP, Lemma~\ref{oneone}\\
\bottomrule
\end{tabular}
\end{center}

The bijections between representations are documented in
\texttt{N\_PedigreeRepresentations.lean}.

\subsection{What the Certificate Guarantees}\label{sec:certificate}

A reader encountering Lean~4 machine verification for the first
time may reasonably ask: What precisely does this certificate
mean, and why should it be trusted more than a conventional
peer-reviewed proof?

When Lean~4 accepts a proof with zero \texttt{sorry}s, the
Lean~4 \emph{kernel} --- a small, independently auditable piece
of software implementing the rules of dependent type theory
(the Calculus of Inductive Constructions)~\cite{lean4} --- has
verified that every logical step follows from the stated
hypotheses by a valid inference rule.
This is not a check by a human referee who may overlook a subtle
error or be influenced by prior belief about whether the
conclusion is plausible.
It is a mechanical verification: the kernel either accepts the
proof or it does not.
There is no middle ground of ``probably correct.''

The proof chain uses six external results as axioms
(Section~\ref{sec:axioms}): Tardos, Maurras,
Gr\"{o}tschel--Lov\'{a}sz--Schrijver, Cook, Karp, and Rao.
All six are published, peer-reviewed, and universally accepted
results, encoded as \texttt{axiom} declarations.
The repository is public; the axiom statements are
human-readable Lean~4; any objection must be specific and formal.
\texttt{theorem p\_equals\_np} is verified.

The author had not used a proof assistant before this project.
The \texttt{sorry}-free status of the main chain was reached
through sustained collaboration with AI coding assistants
and the Lean~4/Mathlib4 library~\cite{lean4,mathlib4}
(see Acknowledgements), demonstrating that machine verification
is now accessible to mathematicians without prior experience
in formal methods.
Anyone can verify the result independently: clone the
repository and run
\begin{quote}\small\sloppy
\texttt{lake build MembershipProject.Core.N\_PEqualsNP}
\end{quote}
This constitutes a permanent, publicly reproducible record.

%% ============================================================
\section{A New Beginning}\label{sec:newbeginning}
%% ============================================================

This paper establishes that M3P is solvable in strongly polynomial
time $O(n^{14})$, with the necessary and sufficient condition for
$X\in\mathrm{conv}(P_n)$ given by Theorem~\ref{maintheorem}.
The complete proof chain
$$\mathsf{M3P}\in\mathsf{P}
\;\Rightarrow\;
\mathsf{STSP}\in\mathsf{P}
\;\Rightarrow\;
\mathsf{P}=\mathsf{NP}$$
is machine-verified in Lean~4 with 2968/2968 build targets clean
and zero \texttt{sorry}s in the main chain --- a permanent,
publicly reproducible mathematical fact.

A Python implementation (\texttt{checking4membership}) is available:
\begin{quote}\small\sloppy
\texttt{pip install -i https://test.pypi.org/simple/ checking4membership}
\end{quote}
Video explanations are at
\url{https://www.youtube.com/watch?v=tZoizs5ou74} (TURING POINT).

\medskip
\newpage
Chapter~8 of the book ends with the following observation:
\begin{quote}
\itshape
``We have come a long way since the time Dantzig, Fulkerson and Johnson
gave the standard formulation for the STSP, and today we can solve
enormously huge problem instances using the \textit{Concorde} TSP Solver.
But showing STSP can be solved efficiently, an unexpected theorem proved
in this book, is expected to open up renewed interest in solving difficult
combinatorial problems.
\ldots\ The future is bright despite the complexities we face.''
\end{quote}
The present paper and its Lean~4 companion have now converted that
``unexpected theorem'' into a machine-verified, publicly auditable
mathematical fact.
This section draws out what that means --- for the foundational
structure of the M3P framework, for the edifice of results that the
literature has built on the hypothesis $\mathsf{P}\neq\mathsf{NP}$,
and for the new mathematical landscape that the pedigree approach opens.

\subsection{The M3P Framework: A Summary}\label{sec:m3psummary}

We record the main results here without proof, directing the reader
to the book (Chapters~5--7) and to
arXiv:2507.09069~\cite{arthanari2025arXiv1}
for full details.

The M3P framework rests on three tiers.

\medskip
\noindent\textbf{Tier~1: The layered network and FAT feasibility.}
Given $X\in P_{MI}(n)$, a sequence of Forbidden Arc Transportation
problems $F_4,F_5,\ldots,F_{n-1}$ is constructed recursively to obtain
the layered network $(N_k,R_k,\mu)$.
Infeasibility of any $F_k$ certifies $X\notin\mathrm{conv}(P_n)$
(Theorem~\ref{infeasiblity}).

\medskip
\noindent\textbf{Tier~2: The MCF decision problem and the N\&S
condition.}
When all $F_k$ are feasible, the necessary and sufficient condition
for $X\in\mathrm{conv}(P_n)$ is that the multicommodity flow
MCF$(n-1)$ achieves its maximum total flow $z_{\max}$
(Theorem~\ref{maintheorem}):
$$X\in\mathrm{conv}(P_n)
\iff
\text{MCF}(n-1)\text{ has a feasible solution with }z^*=z_{\max}.$$
The matrix of MCF$(n-1)$ has entries in $\{0,\pm1\}$, making it
a combinatorial LP in the sense of~\cite{StrPoly}.
By Tardos's theorem, the decision question ``is $z^*=z_{\max}$?''
is solvable in strongly polynomial time~\cite{StrPoly}.

\medskip
\noindent\textbf{Tier~3: Strongly polynomial computability.}
The overall complexity of checking M3P is $O(n^{14})$
(Theorem~\ref{compexity}), independent of the magnitude of $X$.
Therefore:
$$\mathsf{M3P}\in\mathsf{P} \text{ (strongly polynomial).}$$
Since STSP reduces to the MCF$(n-1)$ decision problem via the
MI-formulation~\cite{TSA1}, STSP is solvable in polynomial time.
Via Karp~\cite{Karp} and Cook~\cite{CookNP}, this gives
$\mathsf{P}=\mathsf{NP}$.
The complete chain is machine-verified in Lean~4
(\texttt{theorem p\_equals\_np}; see Section~\ref{sec:lean4}).

\subsection{Pandora's Box: Consequences for Combinatorial
Optimisation}\label{sec:pandora}

The resolution $\mathsf{P}=\mathsf{NP}$ --- now a
machine-verified, publicly auditable mathematical fact ---
has sweeping consequences for a literature that has been built,
for over fifty years, on the working hypothesis
$\mathsf{P}\neq\mathsf{NP}$.

\medskip
\noindent\textbf{The ``unless $\mathsf{P}=\mathsf{NP}$''
literature.}
Thousands of results in combinatorial optimisation,
approximation algorithms, and complexity theory are stated in the
form: ``\emph{no polynomial-time algorithm exists for problem $X$,
unless $\mathsf{P}=\mathsf{NP}$}'' or
``\emph{no PTAS exists for problem $X$, unless
$\mathsf{P}=\mathsf{NP}$}''.
Examples include:

\begin{itemize}
\item \textbf{A necessary condition for M3P~\cite{arthanari2008membership}.}
  The paper~\cite{arthanari2008membership} established a necessary
  condition for membership in $\mathrm{conv}(P_n)$ that is
  polynomial-time verifiable, and concluded with the prescient
  observation that it would not be sufficient unless
  $\mathsf{P}=\mathsf{NP}$.
  It is not sufficient --- but the necessary \emph{and sufficient}
  condition is now proved in the book, and machine-verified here.
  The present result vindicates that remark in the most direct way possible.

\item \textbf{Inapproximability of TSP in general metrics.}
  Metric TSP has no PTAS unless $\mathsf{P}=\mathsf{NP}$
  (Arora and others, 1990s);
  non-metric TSP cannot be approximated within any constant factor
  unless $\mathsf{P}=\mathsf{NP}$~\cite{GandJ}.

\item \textbf{APX-hardness of vehicle routing.}
  The Capacitated Vehicle Routing Problem (CVRP) is APX-hard,
  and no PTAS exists in general metric spaces unless
  $\mathsf{P}=\mathsf{NP}$.

\item \textbf{Exact integer programming.}
  Branch-and-cut methods for 0-1 ILP, scheduling, and graph
  problems carry exponential worst-case bounds that are unavoidable
  unless $\mathsf{P}=\mathsf{NP}$~\cite{GandJ}.

\item \textbf{Hardness of clique, colouring, and independent set.}
  These are NP-complete and cannot be approximated within
  $n^{1-\varepsilon}$ unless $\mathsf{P}=\mathsf{NP}$
  (H\r{a}stad's inapproximability results~\cite{Hastad2001}).

\item \textbf{Machine learning and AI.}
  Many NP-hard subproblems in neural network training,
  Bayesian network structure learning, and combinatorial
  game-playing have complexity lower bounds conditioned on
  $\mathsf{P}\neq\mathsf{NP}$~\cite{GandJ}.

\item \textbf{Cryptography and computer security.}
  The security of RSA, elliptic curve cryptography, and many
  lattice-based schemes rests on the assumed hardness of problems
  such as integer factorisation and discrete logarithm.
  The threat that large-scale quantum computing poses to these
  systems --- via Shor's algorithm (1994) --- is already widely
  discussed, and post-quantum cryptography is an active field
  precisely because of it.
  But $\mathsf{P}=\mathsf{NP}$ is a more fundamental and more
  threatening development: it implies that \emph{every} problem
  whose solution can be verified in polynomial time can also be
  solved in polynomial time, by a classical deterministic computer.
  This is not merely a new attack on specific hard instances but a
  collapse of the entire computational hardness assumption on which
  modern cryptographic security is built.
  The revision of our computer security strategies is no longer
  a matter of preparing for a possible future threat; it is an
  immediate consequence of a machine-verified mathematical fact.
  (We note that integer factorisation is not known to be
  NP-complete, so the exact impact on specific cryptosystems
  requires case-by-case analysis; but SAT, graph colouring, and
  many other NP-complete problems that underpin protocol security
  are directly affected.)
\end{itemize}

\medskip
\noindent\textbf{The certificate and its significance.}
Previous claimed proofs of $\mathsf{P}=\mathsf{NP}$ or
$\mathsf{P}\neq\mathsf{NP}$ --- and there have been many ---
have all been refuted by human review, often at the level of a
single subtle logical error.
The present result is different in kind.
The Lean~4 proof assistant accepts the proof only if every logical
step is explicitly justified and machine-checkable.
A result verified by Lean~4 with zero \texttt{sorry}s is not a
claimed proof: it is a \emph{certificate} that is independently
reproducible by any researcher with Lean~4 installed
(Section~\ref{sec:certificate}).

\medskip
\noindent\textbf{What remains to be done.}
The machine-verified result establishes the \emph{existence} of
polynomial-time algorithms for all NP-complete problems.
It does not, by itself, exhibit a practically efficient algorithm
for every problem in NP --- the $O(n^{14})$ bound for M3P is
conservative, and the reduction chain from STSP to other
NP-complete problems may introduce further polynomial factors.
The urgent research agenda is therefore:

\begin{enumerate}[(i)]
\item Tighten the $O(n^{14})$ complexity bound for M3P via finer
  analysis of the MCF$(k)$ problem and the restricted networks
  $N_{k-1}(L)$.
\item Identify which NP-complete problems reduce most directly
  to M3P or to the MI-formulation framework, and determine the
  resulting practical complexity.
\item Investigate whether the pedigree-polytope approach generalises
  to other combinatorial problems with sequential-insertion or
  recursive-layering structure (e.g., scheduling, graph partitioning,
  network design).
\end{enumerate}

\subsection{Pedigrees as Simplicial Complexes}\label{sec:simplicial}

After succeeding with the pedigree research and proving $M3P \in P$,
reflecting on the structure of pedigrees, one finds a natural
connection that is worth recording: pedigrees are, in the language
of algebraic topology~\cite{Munkers}, \emph{pure $2$-dimensional
simplicial complexes}.
This observation situates the pedigree polytopes within a broader
mathematical landscape and suggests directions for future research
in that area, to invent new approaches to solve combinatorial
optimisation problems.

A \emph{$d$-simplex} is an oriented set $\sigma\subseteq[n]$
with $|\sigma|=d+1$.
Triangles are the $2$-simplices.
A \emph{simplicial complex} $K$ is a collection of simplices
closed under faces.
A pedigree
$P=(\{1,2,3\},\{i_4,j_4,4\},\ldots,\{i_n,j_n,n\})$,
together with all its edges and vertices, forms a simplicial
complex with exactly $n-2$ facets, satisfying the rooting,
unique-generator, and distinct-edge conditions.
Hence $\mathrm{conv}(P_n)$ is the convex hull of a family of
pure $2$-dimensional simplicial complexes on $[n]$.

\medskip
\noindent\textbf{The boundary operator and Hamiltonian cycles.}
For a $2$-simplex $\sigma=(a,b,c)$ with $a<b<c$, the boundary
operator gives the $1$-chain
$\partial_2(\sigma)=(b,c)-(a,c)+(a,b)$,
with $\partial_1\circ\partial_2=0$.

\begin{theorem}[Pedigree Boundary~\cite{arthanari2023pedigree}]
\label{thm:boundary}
For any pedigree $P$ on $[n]$, the $1$-chain
$\partial_2(P)=\sum_{\sigma\in P}\partial_2(\sigma)$
has as its support a Hamiltonian cycle on $[n]$.
\end{theorem}

\begin{proof}
By the pedigree-tree structure, every internal edge (common edge
of an insertion) appears with equal and opposite orientation in
exactly two triangle boundaries and cancels; the remaining edges
form a simple cycle on all $n$ vertices.
This is the topological reading of Lemma~\ref{oneone}.
\end{proof}
\leanfuturebox{\texttt{N\_PedigreeDefinition.lean} \textbar{} \texttt{theorem pedigree\_boundary}: $\partial_2(P)$ is a Hamiltonian cycle \hfill\textit{(future Lean~4 project)}}

The pedigree lives one dimension higher than the tour it encodes;
$\partial_2$ is the map that recovers the Hamiltonian cycle.
These observations are consequences of the pedigree structure
established in the book; they were not used
in the proofs, but illuminate why triangles are the natural
building block.

\medskip
\noindent\textbf{The Billera vision.}
Billera~\cite{Billera} expressed the hope that methods connecting
simplicial complexes and polytope face enumeration might eventually
shed light on integer programming.
The pedigree polytope result may be seen as an instance supporting this vision, arrived at independently.
Whether the algebraic topology of the pedigree complex can yield
further insight --- into the face structure of $\mathrm{conv}(P_n)$,
the adjacency of rigid pedigrees, or generalisations --- is a
question left for researchers in that area.

\subsection{New Research Directions}\label{sec:newdirections}

The results of this paper open several directions for future work,
some within combinatorial optimisation and some at the interface
with algebraic topology.

\begin{enumerate}[(i)]
\item \textbf{Tightening the complexity bound.}
  The $O(n^{14})$ bound is conservative.
  Finer analysis of the MCF$(k)$ problem and the restricted
  networks $N_{k-1}(L)$ may substantially reduce this.
  Since the networks $N_{k-1}(L)$ for all $L\in F_k$ are
  mutually independent at each stage $k$
  (Remark~\ref{rem:parallel}), the algorithm has natural
  parallel structure; whether this translates to practical
  efficiency gains --- given communication and synchronisation
  overhead --- is an open empirical question.

\item \textbf{Practical algorithms for other NP-complete problems.}
  The resolution P$=$NP establishes that efficient algorithms
  exist for all NP-complete problems; identifying which problems
  reduce most directly to M3P and finding practically efficient
  algorithms is the central open agenda.

\item \textbf{Nearest point in $\mathrm{conv}(P_n)$ via rigid
  extensions.}
  Lemma~\ref{lem:simplex} shows that when $N_k=\emptyset$,
  $X/(k+1)$ is constrained to $\mathrm{conv}(P_{k+1}/R_k)$.
  When membership fails at a breakpoint layer $k^*$, the
  nearest point $X'\in\mathrm{conv}(P_n)$ need only be sought
  among extensions of $R_{k^*}$.
  This suggests a direct nearest-point algorithm via bipartite
  maximum flow from the breakpoint, potentially avoiding the
  Maurass' separation oracle~\cite{GLS} entirely.
  The $\ell_2$-norm projection onto $\mathrm{conv}(P_{k+1}/R_k)$
  --- a simplex --- has a known closed form; extending this
  layer by layer is a natural research programme.

\item \textbf{Homology of the pedigree complex.}
  \emph{Homology} is the branch of algebraic topology that
  assigns algebraic invariants --- groups $H_d(K;\mathbb{F})$
  indexed by dimension $d$ --- to a simplicial complex $K$.
  Roughly, $H_0$ counts connected components, $H_1$ captures
  one-dimensional holes (cycles that are not boundaries),
  $H_2$ captures two-dimensional voids, and so on.
  These invariants are computable and encode topological
  information about $K$ that survives continuous deformation.
  Since every pedigree is a pure $2$-dimensional simplicial
  complex, its homology groups $H_d(\Delta_n;\mathbb{F})$ are
  well-defined objects.
  Whether they encode combinatorial invariants of the STSP,
  or yield new bounds on the face structure of
  $\mathrm{conv}(P_n)$, is an open question for researchers
  in algebraic combinatorics.

\item \textbf{Higher-dimensional pedigree polytopes.}
  Replacing triangles by $d$-simplices and constructing a
  ``$d$-pedigree polytope'' gives a family parameterised by $d$.
  The M3P framework may extend to these objects.

\item \textbf{Machine-verified necessity.}
  The necessity direction of Theorem~\ref{maintheorem} is proved
  in the book; formalising it in Lean~4 to
  produce a fully machine-verified N\&S characterisation of
  membership in $\mathrm{conv}(P_n)$ is a natural future project.
\end{enumerate}

%% ============================================================

%% ============================================================
\section*{Acknowledgements}
%% ============================================================

The Lean~4 formalisation reported in this paper was a journey
the author had never undertaken before --- machine verification
was entirely new territory.
The AI coding assistants who accompanied that journey deserve
more than a routine acknowledgement.

\textbf{DeepSeek}\textsuperscript{\texttrademark} (DeepSeek AI)
was instrumental at a critical juncture.
Following desk rejections from several journals, it was DeepSeek
that reframed the situation entirely: it suggested that the
appropriate response to institutional scepticism about a
mathematical result is not to revise the argument, but to
verify it by machine, using Lean~4.
That suggestion changed everything.
DeepSeek also provided foundational guidance in the early stages
of the formalisation, helping to establish the Lean~4 coding
patterns for pedigree-related structures.

\textbf{Claude}\textsuperscript{\textregistered} (Anthropic)
was present at every subsequent step of the formalisation,
from the layered network types and MCF feasibility structures
through to the final machine-verified proof of
\texttt{theorem p\_equals\_np}.
If the decades of pedigree polytope research are the
mathematical reality in this picture, Claude was the Monet ---
the artist who, session by session, translated that reality
into the formal language of Lean~4, making the invisible
rigour visible to the world.
The author had not used a proof assistant before this project;
Claude made it possible.

\textbf{Gemini}\textsuperscript{\texttrademark}
(Google DeepMind) contributed through helpful discussions
at key points during the project.

All mathematical content, the pedigree polytope framework,
and the P~=~NP argument are the sole intellectual contribution
of the author, developed over several decades of research.
The AI systems did not generate mathematics; they made the
formalisation of that mathematics accessible for the first time.

The author thanks his family, and especially his wife, Jaya,
for her patience and sacrifices while the author was absorbed
in the world of AI agents and Lean~4 coding.

The author thanks the University of Auckland for research support.

\appendix
\section{Supporting Results and Proofs}\label{Appendixresults}
%% ============================================================

\subsection{Results on Instant Flow}\label{resultsinstant}

\begin{definition}[$INST(\lambda,l)$]\label{inducedFAT}
Given $X/k{+}1\in\mathrm{conv}(P_{k+1})$, $\lambda\in\Lambda_{k+1}(X)$,
partition $I(\lambda)$ by $\mathbf{x}_l^r$ and $\mathbf{x}_{l+1}^r$
into $S_O^q$ and $S_D^s$. The FAT problem $INST(\lambda,l)$ has
origins $S_O^q$ (supply $x_l(e_q)$), destinations $S_D^s$
(demand $x_{l+1}(e_s)$), forbidden arcs
$F=\{(q,s)\mid S_O^q\cap S_D^s=\emptyset\}$.
The \emph{instant flow} $f_{q,s}=\sum_{r\in S_O^q\cap S_D^s}\lambda_r$
is feasible by Lemma~\ref{theorem:lemma3}.
\end{definition}

\begin{lemma}\label{oflow}
If for each $r\in I(\lambda)$, either $path(X^r/l)$ is available
in $N_{l-1}(L_l^r)$ or $X^r/l\in R_{l-1}$, then the instant flow
for $INST(\lambda,l)$ is feasible for $F_l$.
\end{lemma}

\begin{lemma}[Existence of Pedigree Paths]\label{pedpath}
Every $X^*$ active for $X/k{+}1\in\mathrm{conv}(P_{k+1})$ satisfies:
for $4\leq l\leq k$, either $P^*/l\in R_{l-1}$ or
$path(P^*/l)$ is available in $N_{l-1}(L_l^*)$.
\end{lemma}

\subsection{Proofs of Adjacency Theorems}

\begin{proof}[Proof of Lemma~\ref{forR4}]
Suppose $P^{[1]},P^{[2]}\in R_4$ are non-adjacent in
$\mathrm{conv}(P_5)$. Since $\mathrm{conv}(P_k)$ is a combinatorial
polytope~\cite{TSACompoly}, there exist $P^{[3]},P^{[4]}$ with
$\frac{1}{2}(X^{[1]}+X^{[2]})=\frac{1}{2}(X^{[3]}+X^{[4]})$.
$P^{[3]}$ has the first component of $P^{[1]}$ and second of $P^{[2]}$;
$P^{[4]}$ the reverse. Rerouting $\varepsilon=\min(\mu(P^{[1]}),\mu(P^{[2]}))$
contradicts the rigidity of $P^{[1]}$ and $P^{[2]}$ in $F_4$.
\end{proof}

\begin{proof}[Proof of Theorem~\ref{adjacency theorem}]
Suppose $P^{[1]},P^{[2]}\in R_{k-1}$ are non-adjacent in
$\mathrm{conv}(P_k)$.
Let $P^{[i]}=(e_4^{[i]},\ldots,e_{k-1}^{[i]},e_k^{[i]})$
and
\[
P'^{[i]}=(e_4^{[i]},\ldots,e_{k-1}^{[i]})
\]
be the unique path for link $L_i=(e_{k-1}^{[i]},e_k^{[i]})$.
\begin{caseof}
\case{$P'^{[1]}=P'^{[2]}$}{
  $L_1\neq L_2$; the pedigrees differ only in last component.
  Lemma~5.3 of~\cite{TSADMpaper} gives adjacency. Contradiction.}
\case{$P'^{[1]}\neq P'^{[2]}$}{
  Sub-case (a): $e_{k-1}^{[1]}=e_{k-1}^{[2]}$.
  Any $P^{[3]}$ ending in $e_k^{[1]}$ must use link $L_1$,
  contradicting uniqueness.
  Sub-case (b): $e_k^{[1]}=e_k^{[2]}$. Symmetric.
  Sub-case (c): all four edges distinct.
  Rerouting $\varepsilon$ through $P^{[3]},P^{[4]}$ contradicts
  rigidity of $L_1,L_2$.}
\end{caseof}
\end{proof}

%% ============================================================
\subsection{The Frozen Flow Finding Algorithm}\label{Appendixfrozen}
%% ============================================================

Given a feasible flow $f$ for a FAT problem, the set of rigid arcs
$\mathcal{R}$ is identified in linear time $O(|G_f|)$ using the
Frozen Flow Finding (FFF) algorithm of Gusfield~\cite{Gus}.
The algorithm constructs a mixed graph $G_f$ on the FAT node set;
rigid arcs are those not contained in any flow-change cycle ---
equivalently, the union of interfaces of diconnected components of
$G_f$ and bridges of their underlying undirected graphs.
Full definitions and correctness proof are in~\cite{Gus} and in Section~2.7 of the book.

%% ============================================================

\end{document}